\documentclass[11pt]{article}
\usepackage{amsfonts}
\usepackage{amssymb}
\usepackage{amsmath}
\usepackage{a4wide}
\usepackage{graphicx}

\newtheorem{theorem}{Theorem}

\newtheorem{definition}{Definition}

\newtheorem{lemma}{Lemma}

\newtheorem{proposition}{Proposition}

\newenvironment{proof}[1][Proof]{\textbf{#1.} }{\ \rule{0.5em}{0.5em}}
\makeatletter
\def\@removefromreset#1#2{\let\@tempb\@elt
     \def\@tempa#1{@&#1}\expandafter\let\csname @*#1*\endcsname\@tempa
     \def\@elt##1{\expandafter\ifx\csname @*##1*\endcsname\@tempa\else
    \noexpand\@elt{##1}\fi}     \expandafter\edef\csname cl@#2\endcsname{\csname cl@#2\endcsname}     \let\@elt\@tempb
     \expandafter\let\csname @*#1*\endcsname\@undefined}

\@removefromreset{equation}{section}

\@removefromreset{theorem}{section}
\makeatother
\input{tcilatex}

\begin{document}

\title{On the probabilistic description of a multipartite correlation
scenario with arbitrary numbers of settings and outcomes \\
per site}
\author{Elena R. Loubenets\bigskip \\
Applied Mathematics Department, Moscow State Institute \\
of Electronics and Mathematics, Moscow 109028, Russia}
\maketitle

\begin{abstract}
We consistently formalize the probabilistic description of multipartite
joint measurements performed on systems of any nature. This allows us: (1)
to specify in probabilistic terms the difference between nonsignaling, the
Einstein-Podolsky-Rosen (EPR)\ locality and Bell's locality; (2) to
introduce the notion of an LHV model \ for an $S_{1}\times ...\times S_{N}$%
-setting $N$-partite correlation experiment with outcomes of any spectral
type, discrete or continuous, and to prove both general and specifically
"quantum" statements on an LHV simulation in an arbitrary multipartite case;
(3) to classify LHV models for a multipartite quantum state, in particular,
to show that any $N$-partite quantum state, pure or mixed, admits an
arbitrary $S_{1}\times 1\times ...\times 1$-setting LHV description; (4) to
evaluate a threshold visibility for an arbitrary bipartite noisy quantum
state to admit an $S_{1}\times S_{2}$-setting LHV description under any
generalized quantum measurements of two parties.
\end{abstract}

\tableofcontents

\section{ Introduction}

The probabilistic description of quantum measurements performed by several
parties has been discussed in the literature ever since the seminal
publication [1] of Einstein, Podolsky and Rosen (EPR) in 1935. In that
paper, the authors argued that \emph{locality}\footnote{%
In [1], the Einstein-Podolsky-Rosen locality of parties' measurements is
otherwise expressed as "without in any way disturbing" systems observed by
other parties.} of measurements performed by different parties on perfectly
correlated quantum events implies the "simultaneous reality - and thus
definite values"\footnote{%
See [1], page 778.} of physical quantities described by noncommuting quantum
observables. This EPR argument, contradicting the quantum formalism [2] and
referred to as the EPR paradox, seemed to imply a possibility of a \emph{%
hidden variable} account of quantum measurements. However, the von Neumann
"no-go" theorem [2], published in 1932, was considered wholly to exclude
this possibility.

Analysing this problem in 1964-1966, Bell showed [3] that the setting of von
Neumann "no-go" theorem contains the linearity assumption, which is, in
general, unjustified, and explicitly constructed [3] the hidden variable
(HV) model reproducing the statistical properties of all quantum observables
of a qubit system. Considering, however, spin measurements of two parties on
a two-qubit quantum system in the singlet state, Bell proved [4] that \emph{%
any local} hidden variable (LHV) description of these bipartite measurements
on perfectly correlated quantum events disagrees with the statistical
predictions of quantum theory. Based on his observations in [3, 4], Bell
concluded [3] that the EPR paradox should be resolved specifically due to
the violation of \emph{locality} under multipartite quantum measurements and
that "...non-locality is deeply rooted in quantum mechanics itself and will
persist in any completion"\footnote{%
See [5], page 171.}.

In 1967, Kochen and Specker corrected [6] the setting of von Neumann "no-go"
theorem according to Bell's remark in [3] and proved [6] that, for a quantum
system described by a Hilbert space of a dimension $d\geq 3,$ there does not
exist a non-contextual hidden variable (HV) model that reproduces the
statistical properties of all quantum observables and conserves the
functional subordination between them. Specified for a tensor-product
Hilbert space, the Kochen-Specker theorem excludes the existence of the
non-contextual HV model for \emph{all} projective measurements on a
multipartite quantum state. For \emph{multipartite} projective measurements,
this HV model takes the LHV\ form.

Thus, on one hand, Bell's analysis\footnote{%
In the physical literature, Bell's analysis in [4] is referred to as \emph{%
Bell's theorem}.} in [4] does not exclude a possibility for multipartite
measurements on an\emph{\ arbitrary }nonseparable quantum state to admit an
LHV model. On the other hand, the Kochen-Specker "no-go" theorem [6] \emph{%
does not disprove }the existence for a multipartite quantum state of an LHV
model of a \emph{general} type. Therefore, Bell's analysis [4] plus the
Kochen-Specker theorem [6] do not disprove that multipartite measurements on
an \emph{arbitrary} nonseparable quantum state may admit an LHV model of a 
\emph{general} type.

In 1982, Fine [7] formalized the notion of an LHV model for a bipartite
correlation experiment (not necessarily quantum), with two settings and two
outcomes per site, and proved the main statements on an LHV simulation in
this bipartite case.

In 1989, Werner presented [8] the nonseparable bipartite quantum state on $%
\mathbb{C}^{d}\mathbb{\otimes C}^{d},$ $d\geq 2,$ that admits the LHV model
under any bipartite projective measurements performed on this state.

Ever since these seminal publications, the conceptual and mathematical
aspects of the LHV description of multipartite quantum measurements have
been analysed in a plenty of papers, see, for example, [9-15] and references
therein. The so-called Bell-type inequalities\footnote{%
A Bell-type inequality represents a linear probabilistic constraint (on
either correlation functions or joint probabilities) that holds under any
multipartite correlation experiment admitting an LHV description and may be
violated otherwise.}, specifying multipartite measurement situations
(correlation experiments) admitting an LHV description, are now widely used
in many quantum information tasks.

Nevertheless, as it has been recently noted by Gisin [15], in this field,
there are still "many questions, a few answers".

In our opinion, there\ is even still a lack in a consistent view on \emph{%
locality }under multipartite measurements on spatially separated physical
systems. For example, Werner-Wolf [11] identify \emph{locality} with \emph{%
nonsignaling }while Popescu-Rohrlich [10], Barrett-Linden-Massar-Pironio-
Popescu-Roberts [13] and Masanes-Acin-Gisin [14] specify quantum
multipartite correlations as, in general, \emph{nonlocal }and satisfying 
\emph{"the no-signaling principle"}.\emph{\ }In [12], we argue that, in
contrast to the opinion of Bell in [3, 5], under a multipartite joint
measurement on spacelike separated quantum particles, locality meant by
Einstein, Podolsky and Rosen in [1], \emph{the EPR locality}, is never
violated.

Furthermore, the notion of an LHV model is also understood differently by
different authors. For example, for a bipartite quantum state, Werner's
notion [8] of an LHV model is not equivalent to that of Fine [7] for
bipartite measurements performed on this state.

It should be also stressed that, for an arbitrary multipartite case, there
does not still exist either a consistent analysis of a possibility of an LHV
simulation or a concise analytical approach to the derivation of extreme
Bell-type inequalities for more than two outcomes per site. However, \emph{%
generalized} bipartite quantum measurements on even two qubits may have
infinitely many outcomes.

From the mathematical point of view, the necessity to analyse a possibility
of an LHV simulation arises for any multipartite correlation experiment (not
necessarily quantum), specified not in terms of a single probability space.
The latter is one of the main notions of Kolmogorov's measure-theoretical
formulation [16] of probability theory.

The aim of the present paper is to introduce a consistent frame for the
probabilistic description of a multipartite correlation experiment on
systems of \emph{any} nature and to analyse a possibility of a simulation of
such an experiment in LHV terms. The paper is organized as follows.

In sections 2, 3, we consistently formalize the probabilistic description of
multipartite joint measurements with outcomes of any spectral type, discrete
or continuous, and specify in probabilistic terms the difference between 
\emph{nonsignaling} [17], \emph{the} \emph{EPR locality }[1]\emph{\ }and 
\emph{Bell's locality }[4,\emph{\ }5]. We, in particular, prove (proposition
1) that nonsignaling does not necessarily imply the EPR locality and present
the comparative analysis with the specifications of \emph{locality} and 
\emph{nonsignaling} in [10, 11, 13-15]. The details of the probabilistic
models for the description of EPR local multipartite joint measurements on
physical systems, classical or quantum, are considered in section 3.1.

In section 4, we introduce the notion of an LHV model for an $S_{1}\times
...\times S_{N}$-setting $N$-partite correlation experiment, with outcomes
of any spectral type, discrete or continuous, and prove the general
statements (theorem 1, proposition 2) on an LHV simulation in an arbitrary
multipartite case. An LHV simulation in a general bipartite case and in a
dichotomic multipartite case is considered in theorems 2 and 3, respectively.

In section 5, we classify LHV models arising under EPR local multipartite
joint measurements on a quantum state. We introduce the notion of an $%
S_{1}\times ...\times S_{N}$-setting LHV description of an $N$-partite
quantum state, prove the main general statements (propositions 3 - 6) on
this notion and establish its relation to Werner's notion [8] of an LHV
model for a multipartite quantum state.

The main results of the present paper are summarized in section 6.

\section{Multipartite joint measurements}

Consider a measurement situation where each $n$-th of $N$ parties (players)
performs a measurement, specified by a setting $s_{n}$, and $\Lambda _{n}$
is a set of outcomes $\lambda _{n}$, not necessarily real numbers, observed
by $n$-th party (equivalently, at $n$-th site).

This measurement situation defines the joint\footnote{%
Any measurement with outcomes in a direct product set is called \emph{joint.}%
} measurement with outcomes in $\Lambda _{1}\times ...\times \Lambda _{N}$.
We call this joint measurement $N$-partite and specify it by an $N$-tuple $%
(s_{1},...,s_{N})$ of measurement settings where $n$-th argument refers to a
setting at $n$-th site.

For an $N$-partite joint measurement $(s_{1},...,s_{N})$, denote by\medskip 
\begin{equation}
P_{(s_{1},...,s_{_{N}})}(D_{1}\times ...\times D_{N}):\text{ }=\mathrm{Prob}%
\{\lambda _{1}\in D_{1},...,\text{ }\lambda _{N}\in D_{N}\}  \label{1}
\end{equation}

\noindent the joint probability of events $D_{1}\subseteq \Lambda _{1},$ $%
...,$ $D_{N}\subseteq \Lambda _{N},$ observed by the corresponding parties
and by\footnote{%
For an integral over all values of variables, the domain of integration is
not usually specified.}%
\begin{equation}
\left\langle \Psi (\lambda _{1},...,\lambda _{N})\right\rangle :\text{ }%
=\dint \Psi (\lambda _{1},...,\lambda _{N})\text{ }P_{(s_{1},...,s_{_{N}})}(%
\mathrm{d}\lambda _{1}\times ...\times \mathrm{d}\lambda _{N})  \label{2}
\end{equation}

\noindent the expected value of a bounded measurable real-valued function $%
\Psi (\lambda _{1},...,\lambda _{N}).$ Specified for a function $\Psi $ of
the product form, notation (\ref{2}) takes the form\medskip 
\begin{equation}
\left\langle \varphi _{1}(\lambda _{1})\cdot ...\cdot \varphi _{N}(\lambda
_{N})\right\rangle =\dint \varphi _{1}(\lambda _{1})\cdot ...\cdot \varphi
_{N}(\lambda _{N})\text{ }P_{(s_{1},...,s_{_{N}})}(\mathrm{d}\lambda
_{1}\times ...\times \mathrm{d}\lambda _{N})  \label{3}
\end{equation}%
and may refer either to the joint probability\footnote{%
Here, $\chi _{_{D}}(\lambda ),$ $\lambda \in \Lambda ,$ is an indicator
function of a subset $D\subseteq \Lambda .$ That is: $\chi _{D}(\lambda )=1$
if $\lambda \in D$ and $\chi _{D}(\lambda )=0\ $if\ $\lambda \notin D.$}: 
\begin{eqnarray}
&&\left\langle \text{ }\chi _{D_{1}}(\lambda _{1})\cdot ...\cdot \chi
_{D_{N}}(\lambda _{N})\right\rangle  \label{4} \\
&=&\dint \chi _{D_{1}}(\lambda _{1})\cdot ...\cdot \chi _{D_{N}}(\lambda
_{N})\text{ }P_{(s_{1},...,s_{N})}(\mathrm{d}\lambda _{1}\times ...\times 
\mathrm{d}\lambda _{N})  \notag \\
&=&P_{(s_{1},...,s_{N})}(D_{1}\times ...\times D_{N}),  \notag
\end{eqnarray}%
or, if outcomes are real-valued and bounded, to the mean value:%
\begin{equation}
\left\langle \lambda _{n_{1}}\cdot ...\cdot \lambda _{n_{_{M}}}\right\rangle
=\dint \lambda _{n_{1}}\cdot ...\cdot \lambda _{n_{_{M}}}\text{ }%
P_{(s_{1},...,s_{_{N}})}(\mathrm{d}\lambda _{1}\times ...\times \mathrm{d}%
\lambda _{N})  \label{5}
\end{equation}%
of the product of outcomes observed at $M\leq N$ sites: $1\leq
n_{1}<...<n_{M}\leq N.$ For $M\geq 2,$ the mean value (\ref{5}) is referred
to as \emph{the} \emph{correlation function. }A correlation function\emph{\ }%
for\emph{\ }an $N$-partite joint measurement is called\emph{\ full }whenever%
\emph{\ }$M=N.$

If only outcomes of $M<N$ parties $1\leq n_{1}<...<n_{M}\leq N$ are taken
into account while outcomes of all other parties are ignored then the joint
probability distribution of outcomes observed at these $M$ sites is given by
the following marginal\ 
\begin{equation}
P_{(s_{1},...,s_{_{N}})}(\Lambda _{1}\times ...\times \Lambda
_{n_{1}-1}\times \mathrm{d}\lambda _{_{n_{1}}}\times \Lambda
_{n_{1}+1}\times ...\times \Lambda _{n_{_{M}}-1}\times \mathrm{d}\lambda
_{_{n_{_{_{M}}}}}\times \Lambda _{n_{_{M}}+1}\times ...\times \Lambda
_{_{N}})  \label{6}
\end{equation}%
of distribution $P_{(s_{1},...,s_{_{N}})}.$ In particular, the marginal 
\begin{equation}
P_{(s_{1},...,s_{_{N}})}(\Lambda _{1}\times ...\times \Lambda _{n-1}\times 
\mathrm{d}\lambda _{n}\times \Lambda _{n+1}\times ...\times \Lambda _{_{N}})
\label{7}
\end{equation}%
represents the probability distribution of outcomes observed at $n$-th site.

Recall that events $D_{1},$ $...,D_{_{N}}$ observed by $N$ parties are \emph{%
probabilistically independent }[18] if 
\begin{equation}
P_{(s_{1},...,s_{_{N}})}(D_{1}\times ...\times D_{_{N}})=\text{ }\tprod_{n}%
\text{ }P_{(s_{_{1}},...,s_{_{N}})}(\Lambda _{1}\times ...\times \Lambda
_{n-1}\times D_{n}\times \Lambda _{n+1}\times ...\times \Lambda _{_{N}}).
\label{8}
\end{equation}

\section{Nonsignaling, the EPR locality and Bell's locality}

Consider now an $N$-partite measurement situation where any $n$-th party
performs $S_{n}\geq 1$ measurements, each specified by a positive integer $%
s_{n}\in \{1,...,S_{n}\}.$ Let $\Lambda _{n}^{(s_{n})}$ be a set of outcomes 
$\lambda _{n}^{(s_{n})}$, observed under $s_{n}$-th measurement at $n$-th
site.

This measurement situation ($N$-partite correlation experiment) is described
by the whole family 
\begin{equation}
\mathcal{E}=\{(s_{1},...,s_{N})\mid
s_{1}=1,...,S_{1},...,s_{N}=1,...,S_{N}\},  \label{9}
\end{equation}%
consisting of $S_{1}\times ....\times S_{N}$ joint measurements $%
(s_{1},...,s_{N})$ with joint probability distributions $%
P_{(s_{1},..,s_{_{N}})}^{(\mathcal{E})}$ that may, in general, depend not
only on settings of the corresponding joint measurement $(s_{1},..,s_{_{N}})$
but also on a structure of the whole experiment $\mathcal{E}$, in
particular, on settings of other parties' measurements.

Let, for any joint measurements $(s_{1},...,s_{_{N}}),$ $(s_{1}^{\prime
},...,s_{_{N}}^{\prime })\in \mathcal{E}$, with $M<N$ common settings $%
s_{n_{1}},...,s_{n_{_{M}}}$ at arbitrary sites $1\leq n_{1}<$ $...<n_{M}\leq
N,$ the marginal probability distributions (\ref{6}) of outcomes observed at
these sites coincide, that is:\medskip 
\begin{eqnarray}
&&P_{(s_{1},..,s_{_{N}})}^{(\mathcal{E})}(\Lambda _{1}^{(s_{1})}\times
...\times \Lambda _{n_{1}-1}^{(s_{n_{1}-1})}\times \mathrm{d}\lambda
_{n_{1}}^{(s_{n_{1}})}\times ...\times \mathrm{d}\lambda
_{n_{_{M}}}^{(s_{_{n_{_{M}}}})}\times \Lambda
_{n_{_{M}}+1}^{(s_{n_{M}+1})}\times ...\times \Lambda _{_{N}}^{(s_{N})}) 
\notag \\
&&  \label{11} \\
&=&P_{(s_{1}^{\prime },..,s_{_{N}}^{\prime })}^{(\mathcal{E})}(\Lambda
_{1}^{(s_{1}^{\prime })}\times ...\times \Lambda
_{n_{1}-1}^{(s_{n_{1}-1}^{\prime })}\times \mathrm{d}\lambda
_{n_{1}}^{(s_{n_{1}})}\times ...\times \mathrm{d}\lambda
_{n_{_{M}}}^{(s_{_{n_{_{M}}}})}\times \Lambda
_{n_{_{M}}+1}^{(s_{n_{M}+1}^{\prime })}\times ...\times \Lambda
_{_{N}}^{(s_{N}^{\prime })}).  \notag
\end{eqnarray}

If parties' measurements are performed on \emph{spatially separated} \emph{%
physical} \emph{systems} then (\ref{11}) constitutes a necessary condition
for \emph{nonsignaling} in the sense that: (i) a measurement device of each
party does not directly affect physical systems and measurement devices at
other sites; (ii) spatially separated physical systems either do not
interact with each other or interact \emph{locally}\footnote{%
In the sense that the physical principle of local action [17] is not
violated.} with interaction signals\footnote{%
Interaction signals between physical systems cannot propagate faster than
light.} coming from one system to another already after measurements upon
them. If observed physical systems interact during measurements nonlocally
then the nonsignaling condition (\ref{11}) is, in general, violated.

For a general multipartite correlation experiment, we use a similar
terminology.

\begin{definition}
For a family (\ref{9}) of $N$-partite joint measurements, we refer to (\ref%
{11}) as the nonsignaling condition.
\end{definition}

Let further a measurement of each party be \emph{local} in the Einstein,
Podolsky and Rosen \emph{(EPR)} sense [1]. As specified in footnote 1, the
latter means that results of this measurement are not "in any way disturbed"%
\emph{\ }[1]\emph{\ }by\emph{\ }measurements performed by other parties.

In probabilistic terms, \emph{the EPR locality} of all parties' measurements
under a joint measurement $(s_{1},...,s_{_{N}})\in \mathcal{E}$ is expressed%
\footnote{%
For a bipartite case, this definition was introduced in [12].} by the
dependence of distribution $P_{(s_{1},...,s_{_{N}})}^{(\mathcal{E})}$ and
all its marginals (\ref{6}) only on settings of the corresponding
measurements at the corresponding sites, that is, by the relation:\medskip 
\begin{eqnarray}
&&P_{(s_{1},...,s_{_{N}})}^{(\mathcal{E})}(\Lambda _{1}^{(s_{1})}\times
...\times \Lambda _{n_{1}-1}^{(s_{n_{1}-1})}\times \mathrm{d}\lambda
_{_{n_{1}}}^{(s_{n_{1}})}\times ...\times \mathrm{d}\lambda
_{n_{_{M}}}^{(s_{_{n_{_{M}}}})}\times \Lambda
_{n_{_{M}}+1}^{(s_{n_{M}+1})}\times ...\times \Lambda _{_{N}}^{(s_{N})}) 
\notag \\
&\equiv &P_{(s_{n_{1}},...,s_{n_{_{M}}})}(\mathrm{d}\lambda
_{_{n_{1}}}^{(s_{n_{1}})}\times ...\times \mathrm{d}\lambda
_{_{n_{_{M}}}}^{(s_{n_{_{M}}})}),  \label{12}
\end{eqnarray}%
holding for any $1\leq n_{1}<...<n_{M}\leq N$ and any $1\leq M\leq N.$

With respect to an $N$-partite joint measurement, relation (\ref{12})
induces the following general notion.

\begin{definition}
An $N$-partite joint measurement $(s_{1},...,s_{N})\in \mathcal{E}$ \ is EPR
local if its joint probability distribution has the form $%
P_{(s_{1},...,s_{_{N}})}^{(\mathcal{E})}\equiv P_{(s_{1},...,s_{_{N}})}$ and
all marginals of $P_{(s_{1},...,s_{_{N}})}$ satisfy condition (\ref{12}).
\end{definition}

Note that condition (\ref{12}) does not imply the product form of
distribution $P_{(s_{1},...,s_{N})}.$ Therefore, under an EPR local
multipartite joint measurement, events observed at different sites \emph{do
not need} to be probabilistically independent.

For an EPR local $N$-partite joint measurement $(s_{1},...,s_{_{N}})$, the
marginal probability distribution (\ref{7}) of outcomes observed at $n$-th
site is determined only by a measurement $s_{n}$ at this site and we further
denote it by\medskip 
\begin{equation}
P_{n}^{(s_{n})}(\mathrm{d}\lambda _{n}^{(s_{n})}):\text{ }%
=P_{(s_{1},...,s_{_{N}})}(\Lambda _{1}^{(s_{1})}\times ...\times \Lambda
_{n-1}^{(s_{n-1})}\times \mathrm{d}\lambda _{n}^{(s_{n})}\times \Lambda
_{n+1}^{(s_{n+1})}\times ...\times \Lambda _{_{N}}^{(s_{N})}).  \label{13}
\end{equation}

From (\ref{12}) it follows that any family of EPR local $N$-partite joint
measurements satisfies \emph{the} \emph{nonsignaling} \emph{condition} (\ref%
{11}). However, the converse of this statement is not, in general, true.

\begin{proposition}
For a family (\ref{9}) of $N$-partite joint measurements satisfying the
nonsignaling condition (\ref{11}), each of joint measurements does not need
to be EPR local.
\end{proposition}

\begin{proof}
Consider, for example, the family $\mathcal{E}^{\prime }=\{(a_{i},b_{k})\mid
i,k=1,2\}$ of bipartite\footnote{%
In quantum information literature, two parties are\ traditionally named as
Alice and Bob and their measurements are usually labeled by $a_{i}$ and $%
b_{k}$.} joint measurements, with two settings at each site and the joint
probability distributions\footnote{%
This family of bipartite joint measurements was introduced in [12].}\emph{%
\medskip } 
\begin{equation}
P_{(a_{i},b_{k})}^{(\mathcal{E}^{\prime })}(\mathrm{d}\lambda
_{1}^{(a_{i})}\times \mathrm{d}\lambda _{2}^{(b_{k})})=\dint\limits_{\Omega
}P_{1}^{(a_{i})}(\mathrm{d}\lambda _{1}^{(a_{i})}|\omega )\text{ }%
P_{2}^{(b_{k})}(\mathrm{d}\lambda _{2}^{(b_{k})}|\omega )\text{ }\tau
_{a_{1},a_{2}}^{(b_{1},b_{2})}(\mathrm{d}\omega ),\text{ \ \ \ }i,k=1,2,
\label{14}
\end{equation}%
where measure $\tau _{a_{1},a_{2}}^{(b_{1},b_{2})}$ depends on all
measurements at both parties. From relations%
\begin{eqnarray}
P_{(a_{i},b_{1})}^{(\mathcal{E}^{\prime })}(\mathrm{d}\lambda
_{1}^{(a_{i})}\times \Lambda _{2}^{(b_{1})}) &=&P_{(a_{i},b_{2})}^{(\mathcal{%
E}^{\prime })}(\mathrm{d}\lambda _{1}^{(a_{i})}\times \Lambda _{2}^{(b_{2})})
\label{15} \\
&=&\dint\limits_{\Omega }P_{1}^{(a_{i})}(\mathrm{d}\lambda
_{1}^{(a_{i})}|\omega )\text{ }\tau _{a_{1},a_{2}}^{(b_{1},b_{2})}(\mathrm{d}%
\omega ),\text{ \ \ }\forall i=1,2,  \notag
\end{eqnarray}%
and%
\begin{eqnarray}
P_{(a_{1},b_{k})}^{(\mathcal{E}^{\prime })}(\Lambda _{1}^{(a_{1})}\times 
\mathrm{d}\lambda _{2}^{(b_{k})}) &=&P_{(a_{2},b_{k})}^{(\mathcal{E}^{\prime
})}(\Lambda _{1}^{(a_{2})}\times \mathrm{d}\lambda _{2}^{(b_{k})})
\label{16} \\
&=&\dint\limits_{\Omega }P_{2}^{(b_{k})}(\mathrm{d}\lambda
_{2}^{(b_{k})}|\omega )\text{ }\tau _{a_{1},a_{2}}^{(b_{1},b_{2})}(\mathrm{d}%
\omega ),\text{ \ \ }\forall k=1,2,  \notag
\end{eqnarray}%
it follows that marginals of $P_{(a_{i},b_{k})}^{(\mathcal{E}^{\prime })},$ $%
i,k=1,2,$ satisfy the nonsignaling condition (\ref{11}), though do not, in
general, need to satisfy the EPR locality condition (\ref{12}).
\end{proof}

\bigskip

For an $N$-partite joint measurement $(s_{1},...,s_{N})$ performed on
spatially separated physical systems, \emph{the EPR locality }corresponds to 
\emph{nonsignaling plus no-feedback} of performed measurements on a state of
a composite physical system before all of parties' measurements.

Along with the nonsignaling condition (\ref{11}) and the EPR locality (\ref%
{12}), let us also specify in probabilistic terms the concept of \emph{Bell's%
} \emph{locality,} introduced in [4, 5] for a family of multipartite joint
measurements performed on an identically prepared composite physical system
consisting of spacelike separated particles. This type of locality
corresponds to \emph{nonsignaling plus no-feedback} \emph{plus} \emph{the
existence of variables }$\omega \in \Omega $ of a composite system such that
whenever this system is initially characterized by a variable $\omega \in
\Omega $ with certainty, then, under each joint measurement $%
(s_{1},...,s_{_{N}})\in \mathcal{E},$ any events observed at different sites
are \emph{probabilistically independent:}%
\begin{equation}
P_{(s_{1},...,s_{_{N}})}(\mathrm{d}\lambda _{1}^{(s_{1})}\times ...\times 
\mathrm{d}\lambda _{N}^{(s_{N})}\mid \omega )=P_{1}^{(s_{1})}(\mathrm{d}%
\lambda _{1}^{(s_{1})}|\omega )\cdot ...\cdot P_{N}^{(s_{_{N}})}(\mathrm{d}%
\lambda _{N}^{(s_{N})}|\omega ),\text{ \ \ \ }\forall \omega \in \Omega .
\label{17}
\end{equation}%
If a composite system is initially specified by a probability distribution $%
\nu $ of variables $\omega \in \Omega $ then (\ref{17}) and the law of total
probability\footnote{%
See, for example, in [18].} imply:%
\begin{equation}
P_{(s_{1},...,s_{_{N}})}(\mathrm{d}\lambda _{1}^{(s_{1})}\times ...\times 
\mathrm{d}\lambda _{N}^{(s_{N})})=\dint\limits_{\Omega }P_{1}^{(s_{1})}(%
\mathrm{d}\lambda _{1}^{(s_{1})}|\omega )\cdot ...\cdot P_{N}^{(s_{_{N}})}(%
\mathrm{d}\lambda _{N}^{(s_{N})}|\omega )\text{ }\nu (\mathrm{d}\omega ).
\label{18}
\end{equation}

For a general family of $N$-partite joint measurements, this concept induces
the following notion.

\begin{definition}
A family (\ref{9}) of $N$-partite joint measurements is Bell local if any of
its joint probability distributions admits representation (\ref{18}) where a
probability distribution $\nu $ does not depend on performed measurements.
\end{definition}

From (\ref{11}), (\ref{12}), (\ref{18}) and proposition 1 it follows that,
for an $N$-partite correlation experiment,%
\begin{equation}
Bell^{\prime }s\text{ }locality\text{ }\Rightarrow \text{ }EPR\text{ }%
locality\text{ }\Rightarrow \text{ }Nonsignaling.  \label{19}
\end{equation}%
The converse implications are not, in general, true.

The relation (\ref{19}) between the type of locality meant by Einstein,
Podolsky and Rosen in [1] and the type of locality argued by Bell [4, 5]
indicates that, in contrast to the opinion of Bell [3, 5], the EPR paradox
[1] cannot be, \emph{in principle}, resolved via the violation of Bell's
locality. Moreover, as it is shown in section 3.1, under a multipartite
joint measurement on spacelike separated quantum particles, the EPR locality
is not violated.

Let us now analyse the specification of \emph{locality} and \emph{%
nonsignaling} by other authors.

Werner and Wolf [11] identify "locality" with "nonsignaling" and define it
by the combination of the nonsignaling condition (\ref{11}) with the EPR
locality condition (\ref{12}), specified for a bipartite case. Thus,
Werner-Wolf 's locality [11] constitutes the EPR locality.

Popescu-Rohrlich's [10] "relativistic causality" (nonsignaling) constitutes
the EPR locality (\ref{12}). Barrett-Linden-Massar-Pironio-Popescu-Roberts's
[13] "nonsignaling boxes" correspond to EPR local multipartite correlation
experiments. In both papers [10, 13], "nonlocality" is defined via the
violation of a Bell-type inequality (see footnote 5 and section 4).
Masanes-Acin-Gisin [14] and Gisin [15] define "nonsignaling" and
"nonlocality" similarly to [13].

To our knowledge, the difference (\ref{19}) between nonsignaling [17], the
EPR locality [1] and Bell's locality [4, 5] has not been earlier specified
in the literature.

We stress that the so-called "quantum nonlocality", discussed in the
physical literature ever since the seminal publications [3-5] of Bell, does
not constitute the violation of locality of quantum interactions - under a
multipartite joint measurement on spacelike separated quantum particles, 
\emph{locality of quantum interactions is not violated} (see in section 3.1).

\subsection{EPR local physical models}

Consider now the details of the probabilistic models describing \emph{EPR
local} $N$-partite joint measurements, performed on a composite \emph{%
physical} system, classical or quantum.

\textbf{EPR local} \textbf{classical model}\textit{.} Let, under an EPR
local $N$-partite joint measurement, each party perform a measurement on a 
\emph{classical} subsystem. In this case, there always exist variables $%
\theta \in \Theta $ and a probability distribution $\pi $ (a classical
state) of these variables, characterizing a composite classical system
before measurements and such that, for \emph{any} \emph{EPR local} $N$%
-partite joint measurement $(s_{1},...,s_{_{N}})$ on this classical system
in a state $\pi $, the joint probability distribution $%
P_{(s_{1},...,s_{_{N}})}(\cdot |$ $\pi )$ has the form:%
\begin{equation}
P_{(s_{1},...,s_{_{N}})}(\mathrm{d}\lambda _{1}^{(s_{1})}\times ...\times 
\mathrm{d}\lambda _{N}^{(s_{N})}\text{ }|\text{ }\pi )=\dint\limits_{\Theta
}P_{1}^{(s_{1})}(\mathrm{d}\lambda _{1}^{(s_{1})}|\theta )\cdot ...\cdot
P_{N}^{(s_{_{N}})}(\mathrm{d}\lambda _{N}^{(s_{N})}|\theta )\text{ }\pi (%
\mathrm{d}\theta ),  \label{20}
\end{equation}%
where, for a variable $\theta \in \Theta $ defined initially with certainty, 
$P_{n}^{(s_{n})}(\cdot |\theta )$ represents the probability distribution of
outcomes observed under $s_{n}$-th classical measurement at $n$-th site. In (%
\ref{20}), \emph{the EPR locality} follows from the independence \emph{%
(no-feedback)} of variables $\theta $ and a state $\pi $ on performed
measurements \emph{plus} the independence \emph{(nonsignaling)} of each
conditional distribution $P_{n}^{(s_{n})}(\cdot |\theta )$ on measurements
of other parties.

Let a classical measurement $s_{n}$ at $n$-th site be \emph{ideal, }that is,
describe without an error a property of a composite classical system existed
before this measurement. On a measurable space\footnote{%
In this pair, $\mathcal{F}_{\Theta }$ is a sigma algebra of subsets of a set 
$\Theta .$ For details, see [18, 19].} $(\Theta ,\mathcal{F}_{\Theta }),$
representing a classical composite system before measurements, any of its
observed properties is described by a measurable function $%
f_{n,s_{n}}:\Theta \rightarrow \Lambda _{n}^{(s_{n})}$. In the ideal case,
distribution $P_{n}^{(s_{n},\text{ \textrm{ideal}})}(\cdot |\theta ),\ $%
standing (\ref{20}), takes the form: 
\begin{equation}
P_{n}^{(s_{n},\text{ \textrm{ideal}})}(D_{n}^{(s_{n})}|\theta )=\chi
_{f_{n,s_{n}}^{-1}(D_{n}^{(s_{n})})}(\theta ),  \label{21}
\end{equation}%
where 
\begin{equation}
f_{n,s_{n}}^{-1}(D_{n}^{(s_{n})})=\left\{ \theta \in \Theta \mid
f_{n,s_{n}}(\theta )\in D_{n}^{(s_{n})}\right\} \in \mathcal{F}_{\Theta }
\label{22}
\end{equation}%
is the preimage of a subset $D_{n}^{(s_{n})}\subseteq \Lambda _{n}^{(s_{n})}$
in $\mathcal{F}_{\Theta }$ under mapping $f_{n,s_{n}}.$ If classical
measurements of all parties are ideal, then substituting (\ref{21}) into (%
\ref{20}), we derive that, under an \emph{ideal classical} \emph{EPR local} $%
N$-partite joint measurement $(s_{1},...,s_{_{N}}),$ the joint probability
distribution $P_{(s_{1},...,s_{_{N}})}^{(\text{\textrm{ideal}})}$ has the
image form:%
\begin{equation}
P_{(s_{1},...,s_{_{N}})}^{(\text{\textrm{ideal}})}(D_{1}^{(s_{1})}\times
...\times D_{N}^{(s_{_{N}})}\text{ }|\text{ }\pi )=\pi \left(
f_{1,s_{1}}^{-1}(D_{1}^{(s_{1})})\cap ...\cap
f_{N,s_{_{N}}}^{-1}(D_{N}^{(s_{_{N}})})\right) .  \label{23}
\end{equation}%
\medskip

\textbf{EPR local quantum model.}\emph{\ }If an EPR local $N$-partite joint
measurement is performed on a \emph{quantum }$N$-partite system, then this
system is initially specified by a density operator $\rho $ (a quantum
state) on a complex separable Hilbert space $\mathcal{H}_{1}\otimes
...\otimes \mathcal{H}_{N}$ and, for any \emph{EPR local} $N$-partite joint
measurement performed on this system in a state $\rho $, the joint
probability distribution $P_{(s_{1},...,s_{_{N}})}(\cdot |\rho )$ is given
by:%
\begin{equation}
P_{(s_{1},...,s_{_{N}})}(\mathrm{d}\lambda _{1}^{(s_{1})}\times ...\times 
\mathrm{d}\lambda _{N}^{(s_{N})}\text{ }|\text{ }\rho )=\mathrm{tr}[\rho \{%
\mathrm{M}_{1}^{(s_{1})}(\mathrm{d}\lambda _{1}^{(s_{1})})\otimes ...\otimes 
\mathrm{M}_{N}^{(s_{_{N}})}(\mathrm{d}\lambda _{N}^{(s_{N})})\}],  \label{24}
\end{equation}%
where $\mathrm{M}_{n}^{(s_{n})}(\mathrm{d}\lambda _{n}^{(s_{n})})$ is a
positive operator-valued \emph{(POV)} measure\footnote{$\mathrm{M}%
_{n}^{(s_{n})}$ is a normalized measure with values $\mathrm{M}%
_{n}^{(s_{n})}(D_{n}^{(s_{n})}),$ $\forall D_{n}^{(s_{n})}\subseteq \Lambda
_{n}^{(s_{n})},$ that are positive operators on a complex separable Hilbert
space $\mathcal{H}_{n}$. On the notion of a POV measure, see, for example,
the review section in [20].
\par
{}}, describing $s_{n}$-th quantum measurement at $n$-th site. In (\ref{24}%
), the \emph{EPR locality} is expressed by the independence \emph{%
(no-feedback)} of state $\rho $ on performed measurements \emph{plus} the
independence \emph{(nonsignaling)} of each $\mathrm{M}_{n}^{(s_{n})}\ $on
measurements at other sites.

If $s_{n}$-th measurement of $n$-th party is \emph{ideal,} that is,
reproduces without an error a real-valued quantum property described on $%
\mathcal{H}_{n}$ by a quantum observable $W_{s_{_{n}}},$ then the
corresponding POV measure $\mathrm{M}_{n}^{(s_{n})}$ is projection-valued
and is given by the spectral measure $\mathrm{E}_{_{W_{s_{_{n}}}}}$ of
observable $W_{s_{_{n}}}.$

Let, for example, an $N$-partite joint measurement be performed on \emph{%
spacelike} separated quantum particles in a state $\rho $ on $\mathcal{H}%
_{1}\otimes ...\otimes \mathcal{H}_{N}$. Then its joint probability
distribution has the form (\ref{24}), satisfying the EPR locality condition (%
\ref{12}).

Thus, \emph{under} \emph{any multipartite joint measurement on spacelike
separated quantum particles, the EPR locality (hence, nonsignaling) is not
violated.}

\section{LHV simulation}

Consider a possibility of a local hidden variable (LHV) simulation of an $N$%
-partite correlation experiment described by the $S_{1}\times ...\times
S_{N} $-setting family%
\begin{equation}
\mathcal{E}=\{(s_{1},...,s_{N})\mid s_{1}=1,..,S_{1},\text{ }...,\text{ }%
s_{N}=1,...,S_{N}\},  \label{25}
\end{equation}%
of $N$-partite joint measurements with joint probability distributions 
\begin{equation}
\{P_{(s_{1},...,s_{_{N}})}^{(\mathcal{E})},\text{ \ \ }s_{1}=1,...,S_{1},%
\text{ }...,\text{ }s_{N}=1,...,S_{N}\}.  \label{26}
\end{equation}

The following notion generalizes to an arbitrary multipartite case the
concept of a stochastic hidden variable model, formulated by Fine [7] for a
bipartite case with two settings and two outcomes per site.

\begin{definition}
An $S_{1}\times ...\times S_{N}$-setting family (\ref{25}) of $N$-partite
joint measurements, with $S_{1}+...+S_{N}>N$ and outcomes of any spectral
type, discrete or continuous, admits a local hidden variable\footnote{%
This terminology has been formed historically.} (LHV) model if all its joint
probability distributions (\ref{26}) admit the factorizable representation
of the form:\medskip 
\begin{equation}
P_{(s_{1},...,s_{_{N}})}^{(\mathcal{E})}(\mathrm{d}\lambda
_{1}^{(s_{_{1}})}\times ...\times \mathrm{d}\lambda
_{N}^{(s_{_{N}})})=\dint\limits_{\Omega }P_{1}^{(s_{_{1}})}(\mathrm{d}%
\lambda _{1}^{(s_{_{1}})}|\omega )\cdot ...\cdot P_{N}^{(s_{_{N}})}(\mathrm{d%
}\lambda _{N}^{(s_{_{N}})}|\omega )\text{ }\nu _{\mathcal{E}}(\mathrm{d}%
\omega ),  \label{27}
\end{equation}%
\medskip in terms of a single probability space\footnote{%
In this triple, this triple, $\nu $ is a probability distribution on a
measurable space $(\Omega ,\mathcal{F}_{\Omega })$ (see footnote 15). In
measure theory, triple $(\Omega ,\mathcal{F}_{\Omega },\nu )$ called a
measure space.} $(\Omega ,\mathcal{F}_{\Omega },\nu _{\mathcal{E}})$ and
conditional probability distributions\footnote{%
For any subset $D\subseteq \Lambda ,$ function $P(D|\cdot ):$ $\Omega
\rightarrow \lbrack 0,1]$ is measurable.} $P_{1}^{(s_{_{1}})}(\cdot |\omega
),...,$ $P_{N}^{(s_{_{N}})}(\cdot |\omega ),$ defined $\nu _{\mathcal{E}}$%
-almost everywhere on $\Omega $ and such that each $P_{n}^{(s_{n})}(\cdot
|\omega )$ depends only on a setting of the corresponding measurement at $n$%
-th site.\newline
If, in addition to (\ref{27}), some distributions $P_{n}^{(s_{n})}(\cdot
|\omega )$ corresponding to different sites are correlated then we refer to
such an LHV model as conditional.
\end{definition}

If every party observes a finite number of outcomes, for example, each $%
\Lambda _{n}^{(s_{n})}=\Lambda =\{\lambda _{1},....,\lambda _{K}\},$ then it
suffices to verify the validity of representation (\ref{27}) only for all
one-point subsets $\{\lambda _{k_{1}}\}\times ...\times \{\lambda
_{k_{_{N}}}\}=\{(\lambda _{k_{1}},...,\lambda _{k_{_{N}}})\}\subset \Lambda
^{N}$.

From the LHV representation (\ref{27}) it follows that any family (\ref{25})
of $N$-partite joint measurements admitting an LHV model satisfies\footnote{%
The converse of this statement is not, in general, true.} the \emph{%
nonsignaling} condition (\ref{11}). We stress that, in an LHV model of a
general type, a probability distribution $\nu _{\mathcal{E}}$ has a purely
simulation character and may depend on measurement settings of all (or some)
parties. Therefore, a family of $N$-partite joint measurements admitting a
general LHV model \emph{does not need} to be either EPR local or Bell local%
\emph{\ }(see section 3).

In view of representations (\ref{20}), (\ref{27}), any $S_{1}\times
...\times S_{N}$-setting family (\ref{25}) of EPR local $N$-partite joint
measurements performed on a classical state $\pi $ on $(\Theta ,\mathcal{F}%
_{\Theta })$ admits the LHV model where the probability space is given by $%
(\Theta ,\mathcal{F}_{\Theta },\pi )$ and does not depend on either numbers
or settings of parties' measurements. This LHV model is of the special, 
\emph{classical,} type. From definition 3 it follows that Bell's locality
[4, 5] of a multipartite correlation experiment is equivalent to the
existence for this experiment of an LHV model of the classical type.

If, however, in an $S_{1}\times ...\times S_{N}$-setting family (\ref{25})
of EPR local $N$-partite joint measurements, each of joint measurements is
performed on a quantum state $\rho $ on $\mathcal{H}_{1}\otimes ...\otimes 
\mathcal{H}_{N}$ then, in view of (\ref{24}), this family does not
necessarily admit an LHV model. Possible types of quantum LHV models and
their relation to Werner's notion [8] of an LHV model for a multipartite
quantum state are considered in section 5.

Let us now specify the following type of an LHV model.

\begin{definition}
An LHV model (\ref{27}), conditional or unconditional, is called
deterministic if there exist measurable functions $f_{n,s_{n}}:$ $\Omega
\rightarrow \Lambda _{n}^{(s_{n})}$ such that, in representation (\ref{27}),
all conditional probability distributions have the special form\footnote{%
Here, $\chi _{f_{n,s_{n}}^{-1}(D_{n}^{(s_{n})})}(\omega )$ is an indicator
function of preimage $f_{n,s_{n}}^{-1}(D_{n}^{(s_{n})})$, see (\ref{22}).}:%
\begin{equation}
P_{n}^{(s_{n})}(D_{n}^{(s_{n})}|\omega )=\chi
_{f_{n,s_{n}}^{-1}(D_{n}^{(s_{n})})}(\omega ),\text{ \ \ \ }\forall
D_{n}^{(s_{n})}\subseteq \Lambda _{n}^{(s_{n})},  \label{28}
\end{equation}%
$\nu _{\mathcal{E}}$ -almost everywhere on $\Omega .$
\end{definition}

In a deterministic LHV model specified by a probability space $(\Omega ,%
\mathcal{F}_{\Omega },\nu _{\mathcal{E}})$, to each variable $\omega \in
\Omega ,$ there corresponds the unique outcome $\lambda
_{n}^{(s_{n})}=f_{n,s_{n}}(\omega )$ for any measurement $s_{n}$ at an $n$%
-th site, and all joint distributions $P_{(s_{1},...,s_{N})}^{(\mathcal{E})}$
have the image form 
\begin{equation}
P_{(s_{1},...,s_{N})}^{(\mathcal{E})}(D_{1}^{(s_{1})}\times ...\times
D_{N}^{(s_{_{N}})})=\nu _{\mathcal{E}}\left(
f_{1,s_{1}}^{-1}(D_{1}^{(s_{1})})\cap ...\cap
f_{N,s_{_{N}}}^{-1}(D_{N}^{(s_{_{N}})})\right) ,  \label{29}
\end{equation}%
for any outcome events $D_{1}^{(s_{_{1}})}\subseteq \Lambda
_{1}^{(s_{_{1}})},$ $...,$ $D_{N}^{(s_{_{N}})}\subseteq \Lambda
_{N}^{(s_{_{N}})}.$ The notion of a deterministic LHV model corresponds to
the description of an $S_{1}\times ...\times S_{N}$-setting multipartite
correlation experiment in the frame of Kolmogorov's model [16].

Let an $S_{1}\times ...\times S_{N}$-setting family (\ref{25}) of $N$%
-partite joint measurements admit an LHV model specified by a probability
space $(\Omega ,\mathcal{F}_{\Omega },\nu _{\mathcal{E}}).$ From the
structure of representation (\ref{27}) and formula (\ref{3}) it follows:

\begin{enumerate}
\item the same LHV model holds for any its $K_{1}\times ...\times K_{N}$%
-setting subfamily of $N$-partite joint measurements, where $K_{1}\leq
S_{1}, $ $...,$ $K_{N}\leq S_{N}$, and for any $S_{n_{1}}\times ...\times
S_{n_{_{M}}}$-setting family%
\begin{equation}
\{(s_{n_{1}},...,s_{n_{_{M}}})\mid
s_{n_{1}}=1,...,S_{n_{1}},...,s_{n_{_{M}}}=1,...,S_{n_{_{M}}}\}  \label{30}
\end{equation}%
\medskip of $M$-partite joint measurements: $1\leq n_{1}<...<n_{M}\leq N$, $%
1\leq M<N$, induced by family (\ref{25});

\item for any measurable bounded real-valued functions $\varphi
_{n}^{(s_{n})}(\lambda _{n}^{(s_{n})}),$ $n=1,...,N,$ the expected value of
their product admits the factorizable representation:%
\begin{equation}
\left\langle \varphi _{1}^{(s_{1})}(\lambda _{1}^{(s_{1})})\cdot ...\cdot
\varphi _{N}^{(s_{_{N}})}(\lambda _{N}^{(s_{_{N}})})\right\rangle _{\mathcal{%
E}}=\dint \Phi _{1}^{(s_{1})}(\omega )\cdot ...\cdot \Phi
_{N}^{(s_{_{N}})}(\omega )\text{ }\nu _{\mathcal{E}}(\mathrm{d}\omega ),
\label{31}
\end{equation}%
\medskip with $\nu _{\mathcal{E}}$-measurable functions $\Phi
_{n}^{(s_{n})}(\omega )=\dint \varphi _{n}^{(s_{n})}(\lambda
_{n}^{(s_{n})})P_{n}^{(s_{_{n}})}(\mathrm{d}\lambda _{n}^{(s_{n})}|\omega )$%
. In a deterministic LHV model, $\Phi _{n}^{(s_{n})}(\omega )=(\varphi
_{n}^{(s_{n})}\circ f_{n,s_{n}})(\omega )$ and, in case of real-valued
outcomes,%
\begin{equation}
\left\langle \lambda _{n_{1}}^{(s_{n_{1}})}\cdot ...\cdot \lambda
_{n_{_{M}}}^{(s_{_{n_{_{M}}}})}\right\rangle _{\mathcal{E}}=\dint
f_{n_{1},s_{n_{1}}}(\omega )\cdot ...\cdot f_{n_{_{M}},s_{n_{_{M}}}}(\omega )%
\text{ }\nu _{\mathcal{E}}(\mathrm{d}\omega ),  \label{31'}
\end{equation}%
where the values of functions $f_{n,s_{n}}$ constitute outcomes under the
corresponding measurements at the corresponding sites.\medskip
\end{enumerate}

The following theorem establishes the mutual equivalence of \emph{four}
different statements on an LHV simulation of a multipartite correlation
experiment. Statements \textrm{(a)-(c) }generalize to an arbitrary
multipartite case, with any number of settings and any spectral type of
outcomes at each site, the corresponding propositions of Fine [7] for a $%
2\times 2$-setting bipartite case with two outcomes per site. Statement 
\textrm{(d)} establishes in a general setting the equivalence between the
existence of an LHV model (\ref{27}) and the existence of the LHV-form
representation (\ref{33}) for the product expectations of the special type.

\begin{theorem}
For an $S_{1}\times ...\times S_{N}$-setting family (\ref{25}) of $N$%
-partite joint measurements, with any spectral type of outcomes at each
site, the following statements are equivalent:\medskip \newline
(a) there exists an LHV model formulated by definition 4;\medskip \newline
(b)\textrm{\ }there exists a deterministic LHV model specified by definition
5;\medskip \newline
(c) there exists a joint probability distribution 
\begin{equation}
\mu _{\mathcal{E}}\text{ }(\mathrm{d}\lambda _{1}^{(1)}\times ...\times 
\mathrm{d}\lambda _{1}^{(S_{1})}\times ....\times \mathrm{d}\lambda
_{N}^{(1)}\times ...\times \mathrm{d}\lambda _{N}^{(S_{N})})  \label{32}
\end{equation}%
that returns all distributions $P_{(s_{1},...,s_{_{N}})}^{(\mathcal{E})}$ of
family (\ref{25}) as marginals;\medskip \newline
(d)\textrm{\ }there exists a probability space $(\Omega ,\mathcal{F}_{\Omega
},\nu _{\mathcal{E}})$ and $\nu _{\mathcal{E}}$-measurable real-valued
functions $\Psi _{n}^{(s_{_{n}})}:\Omega \rightarrow \lbrack -1,1]$ on $%
(\Omega ,\mathcal{F}_{\Omega })$ such that, for any $\pm 1$-valued functions 
$\psi _{n}^{(s_{_{n}})}:\Lambda _{n}^{(s_{n})}\rightarrow \{-1,1\},$ the
LHV-form representation: \medskip 
\begin{equation}
\left\langle \psi _{n_{1}}^{(s_{_{n_{1}}})}(\lambda
_{n_{1}}^{(s_{_{n_{1}}})})\cdot ...\cdot \psi
_{n_{_{M}}}^{(s_{_{n_{_{M}}}})}(\lambda
_{n_{_{M}}}^{(s_{_{n_{_{M}}}})})\right\rangle _{\mathcal{E}}=\dint \Psi
_{n_{1}}^{(s_{_{n_{_{1}}}})}(\omega )\cdot ...\cdot \Psi
_{n_{_{M}}}^{(s_{_{n_{_{M}}}})}(\omega )\text{ }\nu _{\mathcal{E}}(\mathrm{d}%
\omega )  \label{33}
\end{equation}%
\medskip holds for arbitrary%
\begin{equation}
1\leq n_{1}<...<n_{M}\leq N,\text{ \ \ \ }1\leq M\leq N.  \label{34}
\end{equation}%
\medskip
\end{theorem}

\begin{proof}
Implication $\mathrm{(b)\Rightarrow (a)}$ is obvious and implication $%
\mathrm{(a)\Rightarrow (d)}$ follows from property (\ref{31}). Let $\mathrm{%
(a)}$ hold. Then each $P_{(s_{1},...,s_{N})}^{(\mathcal{E})}$ admits
representation (\ref{27}) specified by some probability space $(\Omega
^{\prime },\mathcal{F}_{\Omega ^{\prime }},\nu _{\mathcal{E}}^{\prime })$
and conditional distributions $P_{n}^{(s_{n})}(\cdot |\omega ^{\prime }).$
The joint probability measure%
\begin{equation}
\dint\limits_{\Omega _{\mathcal{E}}^{\prime }}\text{ }\tprod%
\limits_{s_{n},n}P_{n}^{(s_{n})}(\mathrm{d}\lambda _{n}^{(s_{n})}|\omega
^{\prime })\text{ }\nu _{\mathcal{E}}^{\prime }(\mathrm{d}\omega ^{\prime })
\label{35}
\end{equation}%
on $\Lambda _{1}^{(1)}\times ...\times \Lambda _{1}^{(S_{1})}\times
....\times \Lambda _{N}^{(1)}\times ...\times \Lambda _{N}^{(S_{N})}$
returns all distributions $P_{(s_{1},...,s_{N})}^{(\mathcal{E})}$ of family (%
\ref{25}) as marginals. Hence, $\mathrm{(a)\Rightarrow (c)}$.

Suppose that \textrm{(c)} holds. Then each $P_{(s_{1},...,s_{N})}^{(\mathcal{%
E)}}$ represents the corresponding marginal of $\mu _{\mathcal{E}}$ and this
means that, for any events $D_{n}^{(s_{n})}\subseteq \Lambda _{n}^{(s_{n})},$%
\begin{eqnarray}
&&P_{(s_{1},...,s_{N})}^{(\mathcal{E)}}(D_{1}^{(s_{1})}\times ...\times
D_{N}^{(s_{_{N}})})  \label{36} \\
&=&\dint \chi _{D_{1}^{(s_{1})}}(\lambda _{1}^{(s_{1})})\cdot ...\cdot \chi
_{D_{N}^{(s_{N})}}(\lambda _{N}^{(s_{_{N}})})\text{ }\mu _{\mathcal{E}}(%
\mathrm{d}\lambda _{1}\times ...\times \mathrm{d}\lambda _{N}),  \notag
\end{eqnarray}%
where, for short, we denote%
\begin{equation}
\lambda _{n}:=(\lambda _{n}^{(1)},...,\lambda _{n}^{(S_{n})}),\text{ \ \ \ \ 
}\Lambda _{n}:=\Lambda _{n}^{(1)}\times ...\times \Lambda _{n}^{(S_{n})}.
\label{37}
\end{equation}%
Representation (\ref{36}) constitutes a particular case of the LHV
representation (\ref{27}), specified by 
\begin{eqnarray}
\omega &=&(\lambda _{1},...,\lambda _{N}),\text{ \ \ \ \ }\Omega =\Lambda
_{1}\times ...\times \Lambda _{N},  \label{38} \\
\nu _{\mathcal{E}} &=&\mu _{\mathcal{E}},\text{ \ \ \ }%
P_{n}^{(s_{n})}(D_{n}^{(s_{n})}|\omega )=\chi _{D_{n}^{(s_{n})}}(\lambda
_{n}^{(s_{_{n}})}).  \notag
\end{eqnarray}%
and, hence, \textrm{(c)}$\mathrm{\Rightarrow (a).}$ Introducing further
measurable functions $f_{n,s_{n}}:\Omega \rightarrow \Lambda _{n}^{(s_{n})}$%
, defined by the relation $f_{n,s_{n}}(\omega ):=\lambda _{n}^{(s_{n})},$
and noting that\footnote{%
For notation $f_{n,s_{n}}^{-1}(D_{n}^{(s_{n})}),$ see (\ref{22}).} 
\begin{equation}
\chi _{D_{n}^{(s_{n})}}(\lambda _{n}^{(s_{n})})=\chi
_{f_{n,s_{n}}^{-1}(D_{n}^{(s_{n})})}(\omega ),  \label{39}
\end{equation}%
we represent (\ref{36}) in the form:%
\begin{eqnarray}
P_{(s_{1},...,s_{N})}^{(\mathcal{E)}}(D_{1}^{(s_{1})}\times ...\times
D_{N}^{(s_{_{N}})}) &=&\dint\limits_{\Omega }\chi
_{f_{1,s_{1}}^{-1}(D_{1}^{(s_{1})})}(\omega )\cdot ...\cdot \chi
_{f_{N,s_{_{N}}}^{-1}(D_{N}^{(s_{_{N}})})}(\omega )\text{ }\nu _{\mathcal{E}%
}(\mathrm{d}\omega )  \label{40} \\
&=&\nu _{\mathcal{E}}\left( f_{1,s_{1}}^{-1}(D_{1}^{(s_{1})})\cap ...\cap
f_{N,s_{_{N}}}^{-1}(D_{N}^{(s_{_{N}})})\right) .  \notag
\end{eqnarray}%
\medskip This representation for (\ref{36}) and definition 5 mean that $%
\mathrm{(c)\Rightarrow }$ $\mathrm{(b)}$. Thus, we have proved%
\begin{equation}
\mathrm{(a)\Leftrightarrow (b)\Leftrightarrow (c),}\ \ \ \mathrm{%
(a)\Rightarrow (d),}  \label{41}
\end{equation}%
and it remains only to show that \textrm{(d) }implies \textrm{(a).}

Consider $\pm 1$-valued functions $\psi _{n}^{(s_{_{n}})}(\lambda
_{n}^{(s_{_{n}})})\in \{-1,1\}$. Let $D_{n}^{(s_{n})}\subseteq \Lambda
_{n}^{(s_{_{n}})}$ be a subset where a function $\psi _{n}^{(s_{_{n}})}$
admits the value $(+1)$. The relation 
\begin{equation}
\psi _{n}^{(s_{_{n}})}(\lambda _{n}^{(s_{_{n}})})=2\chi
_{D_{n}^{(s_{n})}}(\lambda _{n}^{(s_{n})})-1  \label{42}
\end{equation}%
establishes the one-to-one correspondence between $\pm 1$-valued functions $%
\psi _{n}^{(s_{_{n}})}$ on $\Lambda _{n}^{(s_{_{n}})}$ and subsets $%
D_{n}^{(s_{n})}$ $\subseteq \Lambda _{n}^{(s_{_{n}})}$. Due to (\ref{42}),
each $\pm 1$-valued function $\psi _{n}^{(s_{_{n}})}$ on $\Lambda
_{n}^{(s_{_{n}})}$ is uniquely specified by a subset $D_{n}^{(s_{n})}%
\subseteq \Lambda _{n}^{(s_{_{n}})}$ and we replace notation $\psi
_{n}^{(s_{_{n}})}$ $\rightarrow \psi _{D_{n}^{(s_{n})}}.$ Taking (\ref{42})
into account in representation (\ref{4}), we derive:\smallskip 
\begin{equation}
P_{(s_{1},...,s_{_{N}})}^{(\mathcal{E)}}(D_{1}^{(s_{1})}\times ...\times
D_{N}^{(s_{_{N}})})=\frac{1}{2^{N}}\text{ }\left\langle \{1+\psi
_{D_{1}^{(s_{_{1}})}}(\lambda _{1}^{(s_{_{1}})})\}\cdot ...\cdot \{1+\psi
_{D_{N}^{(s_{_{N}})}}(\lambda _{N}^{(s_{_{N}})})\}\right\rangle _{\mathcal{E}%
}.  \label{43}
\end{equation}

Suppose that \textrm{(d) }holds. Then, from representation (\ref{33}) it
follows that, for each $n$ and each $s_{n},$ a correspondence between
functions $\psi _{D_{n}^{(s_{n})}}$ and $\Psi _{n}^{(s_{n})}$ is such that $%
(\Psi _{n}^{(s_{n})}(\Lambda _{n}^{(s_{n})}))(\omega )=1$ and $(\Psi
_{n}^{(s_{n})}(\varnothing ))(\omega )$ $=-1,$ $\nu _{\mathcal{E}}$-almost
everywhere on $\Omega .$

Substituting (\ref{33}) into (\ref{43}), we derive that any joint
distribution $P_{(s_{1},...,s_{N})}$ admits the LHV representation:\ 
\begin{equation}
P_{(s_{1},...,s_{N})}^{(\mathcal{E)}}(D_{1}^{(s_{1})}\times ...\times
D_{N}^{(s_{_{N}})})=\dint\limits_{\Omega
}P_{1}^{(s_{1})}(D_{1}^{(s_{1})}|\omega )\cdot ...\cdot
P_{N}^{(s_{_{N}})}(D_{N}^{(s_{_{N}})}|\omega )\text{ }\nu _{\mathcal{E}}(%
\mathrm{d}\omega ),  \label{44}
\end{equation}%
where 
\begin{equation}
P_{n}^{(s_{n})}(D_{n}^{(s_{n})}\text{ }|\text{ }\omega )=\frac{1}{2}%
\{1+(\Psi _{n}^{(s_{n})}(D_{n}^{(s_{n})}))(\omega )\}.  \label{45}
\end{equation}%
Thus, \textrm{(d)}$\mathrm{\Rightarrow (a).}$ In view of (\ref{41}), this
proves the mutual equivalence of all statements of theorem 1.
\end{proof}

\bigskip

Since different joint probability measures may have the same marginals, in
view of statement \textrm{(c)} of theorem 1, the same multipartite
correlation experiment may admit a few LHV models not reducible to each
other.

Consider a particular $N$-partite case where, say, $n$-th party performs $%
S_{n}\geq 2$ measurements while all other parties perform only one
measurement: $S_{k}=1,$ $k\neq n$. Due to reindexing of sites, any of such
cases is reduced to the $S_{1}\times 1...\times 1$-setting case.

\begin{proposition}
For an arbitrary $S_{1}\geq 2,$ any $S_{1}\times 1...\times 1$-setting
family of $N$-partite joint measurements satisfying the nonsignaling
condition (\ref{11}) admits an LHV model.
\end{proposition}

\begin{proof}
For an $S_{1}\times 1...\times 1$-setting family $\mathcal{E}$ of $N$%
-partite joint measurements, each joint distribution $P_{(s_{1},1,...,1)}^{(%
\mathcal{E)}},$ $s_{1}\in \{1,...,S_{1}\},$ satisfies the relation:%
\begin{equation}
P_{(s_{1},1,...,1)}^{(\mathcal{E)}}(\Lambda _{1}^{(s_{1})}\times D^{\prime
})=0,\text{ \ \ \ \ }\mathrm{\Rightarrow }\text{\ \ \ \ \ }%
P_{(s_{1},1,...,1)}^{(\mathcal{E)}}(D_{1}^{(s_{1})}\times D^{\prime })=0,
\label{46}
\end{equation}%
for any subsets $D_{1}^{(s_{1})}\subseteq \Lambda _{1}^{(s_{1})}$ and $%
D^{\prime }\subseteq \Lambda ^{\prime }=\Lambda _{2}^{(1)}\times ....\times
\Lambda _{N}^{(1)}.$

Implication (\ref{46}) means that, for any subset $D_{1}^{(s_{1})}\subseteq
\Lambda _{1}^{(s_{1})},$ the probability distribution $P_{(s_{1},1,...,1)}^{(%
\mathcal{E})}(D_{1}^{(s_{1})}\times \mathrm{d}\lambda ^{\prime }$ $)$ of
outcomes $\lambda ^{\prime }:=(\lambda _{2}^{(1)},...,\lambda _{N}^{(1)})$
in $\Lambda ^{\prime }$ is absolutely continuous\footnote{%
On this notion and the Radon-Nikodym theorem, see, for example, [18, 19].}
with respect to the marginal $P_{(s_{1},1,...,1)}^{(\mathcal{E)}}(\Lambda
_{1}^{(s_{1})}\times \mathrm{d}\lambda ^{\prime }).$ Therefore, from the
Radon-Nikodym theorem it follows:%
\begin{eqnarray}
&&P_{(s_{1},1,...,1)}^{(\mathcal{E)}}(\mathrm{d}\lambda _{1}^{(s_{1})}\times 
\mathrm{d}\lambda _{2}^{(1)}\times ...\times \mathrm{d}\lambda _{N}^{(1)})
\label{47} \\
&=&\alpha _{s_{1}}^{(\mathcal{E})}(\mathrm{d}\lambda _{1}^{(s_{1})}|\lambda
_{2}^{(1)},...,\lambda _{N}^{(1)})\text{ }P_{(s_{1},1,...,1)}^{(\mathcal{E)}%
}(\Lambda _{1}^{(s_{1})}\times \mathrm{d}\lambda _{2}^{(1)}\times ...\times 
\mathrm{d}\lambda _{N}^{(1)}),  \notag
\end{eqnarray}%
where $\alpha _{s_{1}}^{(\mathcal{E})}(\mathrm{d}\lambda
_{1}^{(s_{1})}|\lambda _{2}^{(1)},...,\lambda _{N}^{(1)})$ is a conditional
probability distribution of outcomes in $\Lambda _{1}^{(s_{1})},$ given a
certain $(\lambda _{2}^{(1)},...,\lambda _{N}^{(1)})\in \Lambda ^{\prime }.$
Since all $N$-partite joint measurements $(s_{1},1,...,1)$ satisfy the
nonsignaling condition (\ref{11}), we have:%
\begin{eqnarray}
&&P_{(s_{1},1,...,1)}^{(\mathcal{E)}}(\Lambda _{1}^{(s_{1})}\times \mathrm{d}%
\lambda _{2}^{(1)}\times ...\times \mathrm{d}\lambda _{N}^{(1)})  \label{48}
\\
&=&P_{(s_{1}^{\prime },1,...,1)}^{(\mathcal{E)}}(\Lambda
_{1}^{(s_{1}^{\prime })}\times \mathrm{d}\lambda _{2}^{(1)}\times ...\times 
\mathrm{d}\lambda _{N}^{(1)})  \notag \\
&\equiv &\tau ^{(\mathcal{E})}(\mathrm{d}\lambda _{2}^{(1)}\times ...\times 
\mathrm{d}\lambda _{N}^{(1)}),\text{ \ \ \ }\forall s_{1},s_{1}^{\prime }\in
\{1,...,S_{1}\}.  \notag
\end{eqnarray}%
The joint probability distribution%
\begin{equation}
\left( \alpha _{1}^{(\mathcal{E)}}(\mathrm{d}\lambda _{1}^{(1)}|\lambda
_{2}^{(1)},...,\lambda _{N}^{(1)})\cdot ...\cdot \alpha _{S_{1}}^{(\mathcal{E%
})}(\mathrm{d}\lambda _{1}^{(S_{1})}|\lambda _{2}^{(1)},...,\lambda
_{N}^{(1)})\right) \tau ^{(\mathcal{E})}(\mathrm{d}\lambda _{2}^{(1)}\times
...\times \mathrm{d}\lambda _{N}^{(1)})  \label{49}
\end{equation}%
returns all distributions $P_{(s_{1},1,...,1)}^{(\mathcal{E)}},$ $%
s_{1}=1,...,S_{1},$ as the corresponding marginals. In view of implication 
\textrm{(c) }$\mathrm{\Rightarrow (a)}$ in theorem 1, this proves the
statement.
\end{proof}

\bigskip

Consider now an LHV simulation of a bipartite correlation experiment.

Due to proposition 2, for an arbitrary $S_{1}\geq 2,$ any $S_{1}\times 1$%
-setting family of bipartite joint measurements satisfying the nonsignaling
condition (\ref{11}) admits an LHV model. The existence of an LHV model for
an arbitrary $S_{1}\times S_{2}$-setting family of bipartite joint
measurements is specified by the following theorem\footnote{%
This theorem generalizes to an \emph{arbitrary} $S_{1}\times S_{2}$-setting
case, with outcomes of any spectral type, Fine's proposition 1 [7, page 292]
for the $2\times 2$-setting case with two outcomes per site.}.

\begin{theorem}
Necessary and sufficient condition for an $S_{1}\times S_{2}$-setting family
of bipartite joint measurements, with outcomes of any spectral type, to
admit an LHV model is the existence of joint probability distributions%
\footnote{%
The lower indices of measures $\mu _{\blacktriangleright }^{(s_{1})}$ on $%
\Lambda _{1}^{(s_{1})}\times \Lambda _{2}^{(1)}\times ...\times \Lambda
_{2}^{(S_{2})}$ and $\mu _{\blacktriangleleft }^{(s_{2})}$ on $\Lambda
_{1}^{(1)}\times ...\times \Lambda _{1}^{(S_{1})}\times \Lambda
_{2}^{(s_{2})}$ indicate a direction of a direct product extension of set $%
\Lambda _{1}^{(s_{1})}\times \Lambda _{2}^{(s_{2})}.$}:%
\begin{equation}
\mu _{\blacktriangleright }^{(s_{1})}(\mathrm{d}\lambda _{1}^{(s_{1})}\times 
\mathrm{d}\lambda _{2}^{(1)}\times ...\times \mathrm{d}\lambda
_{2}^{(S_{2})}),\text{ \ \ }s_{1}=1,...,S_{1},  \label{50}
\end{equation}%
such that each $\mu _{\blacktriangleright }^{(s_{1})}$ returns all
distributions $P_{(s_{1},s_{2})}^{(\mathcal{E)}},$ $s_{2}=1,...,S_{2},$ as
marginals and all $\mu _{\blacktriangleright }^{(s_{1})},$ $%
s_{1}=1,...,S_{1},$ are compatible in the sense that the relation%
\begin{equation}
\mu _{\blacktriangleright }^{(s_{1})}(\Lambda _{1}^{(s_{1})}\times \mathrm{d}%
\lambda _{2}^{(1)}\times ...\times \mathrm{d}\lambda _{2}^{(S_{2})})=\mu
_{\blacktriangleright }^{(s_{1}^{\prime })}(\Lambda _{1}^{(s_{1}^{\prime
})}\times \mathrm{d}\lambda _{2}^{(1)}\times ...\times \mathrm{d}\lambda
_{2}^{(S_{2})})  \label{51}
\end{equation}%
\medskip holds for any $s_{1},$ $s_{1}^{\prime }\in \{1,...,S_{1}\}.$ The
same concerns the existence of joint probability distributions: 
\begin{equation}
\mu _{\blacktriangleleft }^{(s_{2})}(\mathrm{d}\lambda _{1}^{(1)}\times
...\times \mathrm{d}\lambda _{1}^{(S_{1})}\times \mathrm{d}\lambda
_{2}^{(s_{2})}),\text{ \ \ }s_{2}=1,...,S_{2},  \label{52}
\end{equation}%
such that each $\mu _{\blacktriangleleft }^{(s_{2})}$ returns all
distributions $P_{(s_{1},s_{2})}^{(\mathcal{E)}},$ $s_{1}=1,...,S_{1},$ as
marginals and all $\mu _{\blacktriangleleft }^{(s_{2})}$, $%
s_{2}=1,...,S_{2}, $ satisfy the relation:%
\begin{equation}
\mu _{\blacktriangleleft }^{(s_{2})}(\mathrm{d}\lambda _{1}^{(1)}\times
...\times \mathrm{d}\lambda _{1}^{(S_{1})}\times \mathrm{d}\lambda
_{2}^{(s_{2})})=\mu _{\blacktriangleleft }^{(s_{2}^{\prime })}(\mathrm{d}%
\lambda _{1}^{(1)}\times ...\times \mathrm{d}\lambda _{1}^{(S_{1})}\times 
\mathrm{d}\lambda _{2}^{(s_{2}^{\prime })}),  \label{53}
\end{equation}%
\medskip\ for any $s_{2},$ $s_{2}^{\prime }\in \{1,...,S_{2}\}.$
\end{theorem}

\begin{proof}
Denote, for short,%
\begin{equation}
\lambda _{2}:=(\lambda _{2}^{(1)},...,\lambda _{2}^{(S_{2})}),\text{ \ \ \ }%
\Lambda _{2}:=\Lambda _{2}^{(1)}\times ...\times \Lambda _{2}^{(S_{2})}.
\label{54}
\end{equation}%
For each distribution $\mu _{\blacktriangleright }^{(s_{1})}(\mathrm{d}%
\lambda _{1}^{(s_{1})}\times \mathrm{d}\lambda _{2})$ in (\ref{50}), the
relation%
\begin{equation}
\mu _{\blacktriangleright }^{(s_{1})}(\Lambda _{1}^{(s_{1})}\times D_{2})=0%
\text{ \ \ }\mathrm{\Rightarrow }\text{ \ \ }\mu _{\blacktriangleright
}^{(s_{1})}(D_{1}^{(s_{1})}\times D_{2})=0  \label{55}
\end{equation}%
holds for any subsets $D_{1}^{(s_{1})}\subseteq \Lambda _{1}^{(s_{1})}$ and $%
D_{2}\subseteq \Lambda _{2}$. This means that, for any $D_{1}^{(s_{1})}%
\subseteq \Lambda _{1}^{(s_{1})},$ the probability measure $\mu
_{\blacktriangleright }^{(s_{1})}(D_{1}^{(s_{1})}\times \mathrm{d}\lambda
_{2}$ $)$ of outcomes in $\Lambda _{2}$ is absolutely continuous\footnote{%
See reference in footnote 23.} with respect to the marginal probability
distribution $\ \mu _{\blacktriangleright }^{(s_{1})}(\Lambda
_{1}^{(s_{1})}\times \mathrm{d}\lambda _{2}).$ Therefore, each $\mu
_{\blacktriangleright }^{(s_{1})}$admits the Radon-Nikodym representation:%
\begin{equation}
\mu _{\blacktriangleright }^{(s_{1})}(\mathrm{d}\lambda _{1}^{(s_{1})}\times 
\mathrm{d}\lambda _{2})=\alpha _{1}^{(s_{1})}(\mathrm{d}\lambda
_{1}^{(s_{1})}|\lambda _{2})\mu _{\blacktriangleright }^{(s_{1})}(\Lambda
_{1}^{(s_{1})}\times \mathrm{d}\lambda _{2}),  \label{56}
\end{equation}%
where $\alpha _{1}^{(s_{1})}(\cdot |\lambda _{2})$ is a conditional
probability distribution of outcomes $\lambda _{1}^{(s_{1})}\in \Lambda
_{1}^{(s_{1})}.$ In view of (\ref{51}), we denote 
\begin{eqnarray}
\mu _{\blacktriangleright }^{(s_{1})}(\Lambda _{1}^{(s_{1})}\times \mathrm{d}%
\lambda _{2}) &=&\mu _{\blacktriangleright }^{(s_{1}^{\prime })}(\Lambda
_{1}^{(s_{1}^{\prime })}\times \mathrm{d}\lambda _{2})  \label{57} \\
&=&\tau _{2}(\mathrm{d}\lambda _{2}),\text{ \ \ \ \ }s_{1},\text{ }%
s_{1}^{\prime }\in \{1,...,S_{1}\}.  \notag
\end{eqnarray}%
The joint probability measure%
\begin{equation}
\left( \alpha _{1}^{(1)}(\mathrm{d}\lambda _{1}^{(1)}|\lambda _{2})\cdot
...\cdot \alpha _{1}^{(S_{1})}(\mathrm{d}\lambda _{1}^{(S_{1})}|\lambda
_{2})\right) \tau _{2}(\mathrm{d}\lambda _{2})  \label{58}
\end{equation}%
returns all $P_{(s_{1},s_{2})}^{(\mathcal{E)}}$ as marginals. In view of
theorem 1, this proves the sufficiency part of theorem 2.

In order to prove the necessity part, let an $S_{1}\times S_{2}$-setting
family admit a LHV model. Then, by statement \textrm{(c)} of theorem 1,
there exists a joint probability distribution $\mu _{\mathcal{E}}(d\lambda
_{1}^{(1)}\times ...\times d\lambda _{1}^{(S_{1})}\times d\lambda
_{2}^{(1)}\times ...\times d\lambda _{2}^{(S_{2})})$ of all outcomes
observed by two parties. The marginals 
\begin{equation}
\mu _{\mathcal{E}}(\Lambda _{1}^{(1)}\times ..\times \Lambda
_{1}^{(s_{1}-1)}\times \mathrm{d}\lambda _{1}^{(s_{1})}\times \Lambda
_{1}^{(s_{1}+1)}...\times \Lambda _{1}^{(S_{1})}\times \mathrm{d}\lambda
_{2}^{(1)}\times ...\times \mathrm{d}\lambda _{2}^{(S_{2})}),  \label{59}
\end{equation}%
constitute the probability distributions $\mu _{\blacktriangleright
}^{(s_{1})},$ specified by (\ref{50}), (\ref{51}). For measures $\mu
_{\blacktriangleleft }^{(s_{2})},$ the necessity and sufficiency parts are
proved quite similarly.
\end{proof}

\medskip

Theorems 1, 2 and proposition 2 refer to an LHV simulation of an arbitrary
multipartite correlation experiment with outcomes of any spectral type.
Below, we consider peculiarities of an LHV simulation in a multipartite case
with only two outcomes per site.

\subsection{ \ A dichotomic multipartite case}

Let, under an $N$-partite joint measurement $(s_{1},...,s_{N})$, each party
perform a measurement with only two outcomes, that is, a \emph{dichotomic}
measurement. These two outcomes do not need to be numbers, however, due to
possible mappings $\lambda _{n}^{(s_{n})}\mapsto \varphi
_{n}^{(s_{n})}(\lambda _{n}^{(s_{n})})\in \{-1,1\},$ it suffices to analyse
only a dichotomic case with outcomes: $\lambda _{n}^{(s_{n})}=\pm 1.$

Since the direct product $\{\lambda _{1}^{(s_{1})}\}\times ...\times
\{\lambda _{N}^{(s_{N})}\}$ of one-point subsets constitutes the one-point
subset $\{(\lambda _{1}^{(s_{1})},...,\lambda _{N}^{(s_{_{N}})})\}\subset
\Lambda _{1}^{(s_{1})}\times ...\times \Lambda _{N}^{(s_{_{N}})}$, for a
discrete case, we further omit brackets $\{\cdot \}$ and denote:%
\begin{equation}
P_{(s_{_{1}},...,s_{_{N}})}^{(\mathcal{E)}}\text{ }(\{\lambda
_{1}^{(s_{1})}\}\times ...\times \{\lambda _{N}^{(s_{N})}\})\equiv
P_{(s_{_{1}},...,s_{_{N}})}^{(\mathcal{E)}}(\lambda
_{1}^{(s_{1})},...,\lambda _{N}^{(s_{_{N}})}).  \label{60}
\end{equation}%
For a further consideration, we need to prove the following general
statement.

\begin{lemma}
For an arbitrary $N$-partite joint measurement $(s_{1},...,s_{N})\in 
\mathcal{E},$ with $\pm 1$-valued outcomes at each site, 
\begin{eqnarray}
&&2P_{(s_{_{1}},...,s_{_{N}})}^{(\mathcal{E)}}(\lambda
_{1}^{(s_{_{1}})},...,\lambda _{N}^{(s_{_{N}})})  \label{61} \\
&=&1\text{ \ \ }+\sum_{\substack{ 1\leq n_{1}<...<n_{_{N-k}}\leq N,  \\ %
k=0,...,N-1}}\xi (\lambda _{n_{_{1}}}^{(s_{n_{_{1}}})})\cdot ...\cdot \xi
(\lambda _{n_{_{N-k}}}^{(s_{n_{_{N-k}}})})\text{ }\left\langle \lambda
_{n_{_{1}}}^{(s_{n_{_{1}}})}\cdot ...\cdot \lambda
_{n_{_{N-k}}}^{(s_{_{n_{_{N-k}}}})}\right\rangle _{\mathcal{E}},  \notag
\end{eqnarray}%
where $\xi (\pm 1)=\pm 1.$
\end{lemma}

\begin{proof}
Due to relations%
\begin{equation}
2\chi _{_{\{1\}}}(\lambda _{n}^{(s_{n})})-1=\lambda _{n}^{(s_{n})},\text{ \
\ \ \ }2\chi _{_{\{-1\}}}(\lambda _{n}^{(s_{n})})-1=-\lambda _{n}^{(s_{n})},
\label{62}
\end{equation}%
holding for each $\lambda _{n}^{(s_{n})}\in \{-1,1\},$ we have:\medskip 
\begin{equation}
\chi _{_{D_{n}^{(s_{n})}}}(\lambda _{n}^{(s_{n})})=\frac{1+\lambda
_{n}^{(s_{n})}\xi (D_{n}^{(s_{n})})}{2},\text{ \ \ \ }\xi (\{1\})=1,\ \ \
\xi (\{-1\})=-1,  \label{63}
\end{equation}%
\medskip for each of one-point subsets $\{-1\}\ $or $\{1\}$.

Substituting (\ref{63}) into (\ref{4}), for any direct product combination $%
D_{1}^{(s_{1})}\times ...\times D_{N}^{(s_{_{N}})}$ of one-point subsets $%
\{-1\}$ and $\{1\},$we derive:\medskip 
\begin{eqnarray}
&&P_{(s_{_{1}},...,s_{_{N}})}^{(\mathcal{E)}}(D_{1}^{(s_{1})}\times
...\times D_{N}^{(s_{_{N}})})  \label{64} \\
&=&\frac{1}{2^{N}}\text{ }\left\langle (1+\lambda _{1}^{(s_{1})}\xi
(D_{1}^{(s_{1})}))\cdot ...\cdot (1+\lambda _{N}^{(s_{_{N}})}\xi
(D_{N}^{(s_{_{N}})}))\right\rangle _{\mathcal{E}}\text{ }  \notag \\
&=&\frac{1}{2^{N}}+\frac{1}{2^{N}}\sum_{\substack{ 1\leq
n_{1}<...<n_{N-k}\leq N,  \\ k=0,...,N-1}}\xi (D_{n_{1}}^{(s_{n_{1}})})\cdot
...\cdot \xi (D_{n_{_{N-k}}}^{(s_{n_{_{N-k}}})})\text{ }\left\langle \lambda
_{n_{_{1}}}^{(s_{n_{1}})}\cdot ...\cdot \lambda
_{n_{_{N-k}}}^{(s_{_{n_{_{N-k}}}})}\right\rangle _{\mathcal{E}}.  \notag
\end{eqnarray}%
\bigskip Using in (\ref{64}) notation (\ref{60}) and renaming $\xi
(\{1\})\rightarrow \xi (1),$ $\xi (\{-1\})\rightarrow \xi (-1),$ we prove (%
\ref{61}).
\end{proof}

\bigskip

From (\ref{61}) it, in particular, follows:%
\begin{equation}
2^{N}P_{(s_{_{1}},...,s_{_{N}})}^{(\mathcal{E)}}(1,...,1)=1+\sum_{\substack{ %
1\leq n_{1}<...<n_{_{N-k}}\leq N,  \\ k=0,...,N-1}}\left\langle \lambda
_{n_{_{1}}}^{(s_{n_{1}})}\cdot ...\cdot \lambda
_{n_{_{N-k}}}^{(s_{_{n_{_{N-k}}}})}\right\rangle _{\mathcal{E}}.  \label{65}
\end{equation}

In view of lemma 1, the mutual equivalence of statements $\mathrm{(a)}$ and%
\textrm{\ }$\mathrm{(}$\textrm{d) }of theorem 1 takes the following form.

\begin{theorem}
An $S_{1}\times ...\times S_{N}$-setting family (\ref{25}) of $N$-partite
joint measurements, with $\pm 1$-valued outcomes at each site, admits an LHV
model, formulated by definition 4, iff there exist a probability space $%
(\Omega ,\mathcal{F}_{\Omega },\nu _{\mathcal{E}})$ and $\nu _{\mathcal{E}}$%
-measurable real-valued functions 
\begin{equation}
f_{n,s_{n}}:\Omega \rightarrow \mathbb{[}-1,1\mathbb{]},\text{ \ \ \ }%
\forall s_{n},\forall n,  \label{66}
\end{equation}%
on $(\Omega ,\mathcal{F}_{\Omega })$ such that any of the mean values:%
\begin{equation}
\left\langle \lambda _{n_{_{1}}}^{(s_{n_{_{1}}})}\cdot ...\cdot \lambda
_{n_{_{M}}}^{(s_{_{n_{M}}})}\right\rangle _{\mathcal{E}},\text{\ \ \ }1\leq
n_{1}<...<n_{M}\leq N,\text{ \ \ \ }1\leq M\leq N,  \label{67}
\end{equation}%
admits the representation%
\begin{equation}
\left\langle \lambda _{n_{_{1}}}^{(s_{n_{1}})}\cdot ...\cdot \lambda
_{n_{_{M}}}^{(s_{_{n_{M}}})}\right\rangle _{\mathcal{E}}=\dint
f_{n_{_{1}},s_{n_{_{1}},}}(\omega )\cdot ...\cdot
f_{n_{_{_{M}}},s_{n_{_{M}}}}(\omega )\text{ }\nu _{\mathcal{E}}(\mathrm{d}%
\omega )  \label{68}
\end{equation}%
of the LHV-form.
\end{theorem}

\begin{proof}
The necessity follows from property 2 (see formula (\ref{31})). In order to
prove the sufficiency part, let us substitute (\ref{68}) into formula (\ref%
{61}), in the form (\ref{64}). For any direct product combination $%
D_{1}^{(s_{1})}\times ...\times D_{N}^{(s_{_{N}})}$ of one-point subsets $%
\{-1\}$ and $\{1\},$ we derive:%
\begin{eqnarray}
&&P_{(s_{_{1}},...,s_{_{N}})}^{(\mathcal{E)}}(D_{1}^{(s_{1})}\times
...\times D_{N}^{(s_{_{N}})})  \label{69} \\
&=&\frac{1}{2^{N}}\dint \text{ }[1+\xi (D_{1}^{(s_{1})})f_{1,s_{1}}(\omega
)]\cdot ...\mathbf{\cdot }\text{ }[1+\xi
(D_{N}^{(s_{_{N}})})f_{N,s_{_{N}}}(\omega )]\text{ }\nu _{\mathcal{E}}(%
\mathrm{d}\omega ).  \notag
\end{eqnarray}%
Extending (\ref{69}) to all subsets of set $\{-1,1\},$ we have:%
\begin{equation}
P_{(s_{_{1}},...,s_{_{N}})}^{(\mathcal{E)}}(D_{1}^{(s_{1})}\times ...\times
D_{N}^{(s_{_{N}})})=\int P_{1}^{(s_{1})}(D_{1}^{(s_{1})}\text{ }|\text{ }%
\omega )\cdot ...\cdot P_{N}^{(s_{_{N}})}(D_{N}^{(s_{_{N}})}\text{ }|\text{ }%
\omega )\text{ }\nu _{\mathcal{E}}(\mathrm{d}\omega ),  \label{70}
\end{equation}%
where%
\begin{eqnarray}
P_{n}^{(s_{n})}(\{1\}\text{ }|\text{ }\omega ) &=&\frac{1}{2}%
[1+f_{n,s_{n}}(\omega )],\text{ \ \ \ \ }P_{n}^{(s_{n})}(\{-1\}\text{ }|%
\text{ }\omega )=\frac{1}{2}[1-f_{n,s_{n}}(\omega )],  \label{71} \\
P_{n}^{(s_{n})}(\varnothing \text{ }|\text{ }\omega ) &=&0,\text{ \ \ \ \ \ }%
P_{n}^{(s_{n})}(\{-1,1\}\text{ }|\text{ }\omega )=1.  \notag
\end{eqnarray}%
This proves the statement.
\end{proof}

\bigskip

From theorem 3 it follows that, for an arbitrary $S_{1}\times ...\times
S_{N} $-setting family of $N$-partite joint measurements with two outcomes
per site, the existence of the LHV-form representation (\ref{68}) for \emph{%
only} the full correlation functions \emph{does not}, in general, imply the
existence of an LHV model (\ref{27}) for joint probability distributions.

All statements of section 4 refer to an LHV simulation of a general
correlation experiment. In the following section, we specify an LHV
simulation in a quantum multipartite case.

\section{Quantum LHV models}

We start by analysing an LHV simulation of an $S_{1}\times S_{2}$-setting
family of \emph{EPR local} bipartite joint measurements performed on a
separable quantum state: 
\begin{equation}
\rho _{sep}=\sum_{m}\gamma _{m}\rho _{1}^{(m)}\otimes \rho _{2}^{(m)},\text{
\ \ \ }\gamma _{m}\geq 0,\text{ \ \ \ }\sum_{m}\gamma _{m}=1,  \label{72}
\end{equation}%
on a complex separable Hilbert space $\mathcal{H}\otimes \mathcal{H}$,
possibly, infinite dimensional.

Let, at each $n$-th site, quantum measurements be described by POV measures $%
\mathrm{M}_{n}^{(s_{n})}(\mathrm{d}\lambda _{n}^{(s_{1})})$, $%
s_{n}=1,...,S_{n}$, $n=1,2.$ From (\ref{24}) and (\ref{72}) it follows that
this correlation experiment is described by the joint probability
distributions of the form%
\begin{equation}
P_{(s_{1},s_{2})}(\mathrm{d}\lambda _{1}^{(s_{1})}\times \mathrm{d}\lambda
_{2}^{(s_{2})}\text{ }|\text{ }\rho _{sep})=\sum_{m}\gamma _{m}\mathrm{tr}%
[\rho _{1}^{(m)}\mathrm{M}_{1}^{(s_{1})}(\mathrm{d}\lambda _{1}^{(s_{1})})]%
\text{ }\mathrm{tr}[\rho _{2}^{(m)}\mathrm{M}_{2}^{(s_{2})}(\mathrm{d}%
\lambda _{2}^{(s_{2})})].  \label{73}
\end{equation}%
This form constitutes a particular case of the LHV representation (\ref{27}%
), specified by the probability space with elements 
\begin{equation}
\Omega ^{\prime }=\{m=1,2,....\},\ \ \ \ \ \ \nu _{m}^{\prime }=\gamma _{m},%
\text{ \ \ }\forall m\in \Omega ^{\prime },  \label{74}
\end{equation}%
and conditional distributions $P_{n}^{(s_{n})}(\mathrm{\cdot }$ $|m)=\mathrm{%
tr}[\rho _{n}^{(m)}\mathrm{M}_{n}^{(s_{n})}(\mathrm{\cdot })],$ $%
s_{n}=1,...,S_{n},$ $n=1,2,$ for any $m\in \Omega ^{\prime }.$

Thus, any $S_{1}\times S_{2}$-setting family of bipartite joint measurements
performed on a separable quantum state $\rho _{sep}$ admits the LHV model
where the probability space is determined only by this separable state and
does not depend on either numbers or settings of parties' measurements, that
is, the LHV model of the classical type (see section 4).

Furthermore, all $P_{(s_{1},s_{2})}^{(\mathcal{E})}(\mathrm{\cdot }|$ $\rho
_{sep})$, $s_{1}=1,...,S_{1},$ $s_{2}=1,...,S_{2},$ defined by (\ref{73}),
are marginals of the joint probability measure%
\begin{eqnarray}
&&\mu _{\rho _{sep}}(\mathrm{d}\lambda _{1}^{(1)}\times ...\times \mathrm{d}%
\lambda _{1}^{(S_{1})}\times \mathrm{d}\lambda _{2}^{(1)}\times ...\times 
\mathrm{d}\lambda _{2}^{(S_{2})})  \label{75} \\
&=&\sum_{m}\gamma _{m}\dprod\limits_{s_{1},\text{ }s_{2}}\mathrm{tr}[\rho
_{1}^{(m)}\mathrm{M}_{1}^{(s_{1})}(\mathrm{d}\lambda _{1}^{(s_{1})})]\text{ }%
\mathrm{tr}[\rho _{2}^{(m)}\mathrm{M}_{2}^{(s_{2})}(\mathrm{d}\lambda
_{2}^{(s_{2})})].  \notag
\end{eqnarray}%
Therefore, from the proof of implication $\mathrm{(c)\Rightarrow (a)}$ in
theorem 1 (see representation (\ref{36})) it follows that the considered
correlation experiment admits also the LHV model which is specified by the
probability space $(\Omega ,\mathcal{F}_{\Omega },\mu _{\rho _{sep}}),$ with 
\begin{eqnarray}
\omega &=&(\lambda _{1}^{(1)},...,\lambda _{1}^{(S_{1})},\lambda
_{2}^{(1)},...,\lambda _{2}^{(S_{2})}),  \label{76} \\
\Omega &=&\Lambda _{1}^{(1)}\times ...\times \Lambda _{1}^{(S_{1})}\times
\Lambda _{2}^{(1)}\times ...\times \Lambda _{2}^{(S_{2})},  \notag
\end{eqnarray}%
and conditional distributions $P_{n}^{(s_{n})}(D_{n}^{(s_{n})}|\omega )=\chi
_{D_{n}^{(s_{n})}}(\omega ).$ The latter LHV model is induced by the LHV
model (\ref{74}).

Consider further an $S_{1}\times S_{2}$-setting bipartite correlation
experiment, performed on the \emph{specific} bipartite separable state%
\begin{equation}
\widetilde{\rho }_{sep}=\sum_{m}\gamma _{m}|e_{m}\rangle \langle
e_{m}|\otimes |e_{m}\rangle \langle e_{m}|,  \label{77}
\end{equation}%
where $\{e_{m}\}$ is an orthonormal basis in $\mathcal{H}$. Since state $%
\widetilde{\rho }_{sep}$ is reduced from the nonseparable pure state 
\begin{equation}
T=|\text{ }\sum_{m}\sqrt{\gamma _{m}}e_{m}^{\otimes (S_{1}+S_{2}}\text{ }%
\rangle \langle \text{ }\sum_{m}\sqrt{\gamma _{m}}e_{m}^{\otimes
(S_{1}+S_{2})}\text{ }|  \label{78}
\end{equation}%
on $\mathcal{H}^{\otimes (S_{1}+S_{2})},$ all distributions $%
P_{(s_{1},s_{2})}(\cdot $ $|$ $\widetilde{\rho }_{sep})$ represent marginals
of the joint measure\medskip 
\begin{eqnarray}
&&\mu _{\widetilde{\rho }_{sep}}^{\prime }(\mathrm{d}\lambda
_{1}^{(1)}\times ...\times \mathrm{d}\lambda _{1}^{(S_{1})}\times \mathrm{d}%
\lambda _{2}^{(1)}\times ...\times \mathrm{d}\lambda _{2}^{(S_{2})})
\label{79} \\
&=&\mathrm{tr}[T\{\mathrm{M}_{1}^{(1)}(\mathrm{d}\lambda _{1}^{(1)})\otimes
...\otimes \mathrm{M}_{1}^{(S_{1})}(\mathrm{d}\lambda _{1}^{(S_{1})})\otimes 
\mathrm{M}_{2}^{(1)}(\mathrm{d}\lambda _{2}^{(1)})\otimes ...\otimes \mathrm{%
M}_{2}^{(S_{2})}(\mathrm{d}\lambda _{2}^{(S_{2})})\}]  \notag \\
&=&\dsum\limits_{m,l}\sqrt{\gamma _{m}}\sqrt{\gamma _{l}}\dprod%
\limits_{s_{1},s_{2}}\langle e_{m}|\mathrm{M}_{1}^{(s_{1})}(\mathrm{d}%
\lambda _{1}^{(s_{1})})|e_{l}\rangle \text{ }\langle e_{m}|\mathrm{M}%
_{2}^{(s_{2})}(\mathrm{d}\lambda _{2}^{(s_{2})})|e_{l}\rangle .  \notag
\end{eqnarray}%
Quite similarly as explained above, this implies that any $S_{1}\times S_{2}$%
-setting family of bipartite joint measurements performed on $\widetilde{%
\rho }_{sep}$ admits the LHV model, specified by the probability space $%
(\Omega ,\mathcal{F}_{\Omega },\mu _{\widetilde{\rho }_{sep}}^{\prime })$,
where variables $\omega \in \Omega $ are defined by (\ref{76}) while
distribution $\mu _{\widetilde{\rho }_{sep}}\neq \mu _{\widetilde{\rho }%
_{sep}}^{\prime }.$ The latter LHV model is not reducible to the LHV model (%
\ref{74}) of the classical type.

Thus, any $S_{1}\times S_{2}$-setting bipartite correlation experiment
performed on state $\widetilde{\rho }_{sep}$ admits at least two LHV models
not reducible to each other. The first LHV model, with the probability space
(\ref{74}) depending only on state $\widetilde{\rho }_{sep}$, holds for 
\emph{any} setting $S_{1}\times S_{2}.$ The second LHV model, with the
probability space $(\Omega ,\mathcal{F}_{\Omega },\mu _{\widetilde{\rho }%
_{sep}}^{\prime })$, is constructed specifically for a given setting $%
S_{1}\times S_{2}$.

In view of this analysis, we introduce the following notions.

\begin{definition}
An $N$-partite quantum state $\rho $ admits an $S_{1}\times ...\times S_{N}$%
-setting LHV description if any $S_{1}\times ...\times S_{N}$-setting family
of EPR local $N$-partite joint measurements performed on this quantum state
admits an LHV\ model formulated by definition 4.
\end{definition}

This definition and the LHV property 1 (specified in section 3 after
definition 5) imply the following statements on a LHV description of an
arbitrary $N$-partite quantum state.

\begin{proposition}
Let an $N$-partite quantum state $\rho $ on $\mathcal{H}_{1}\otimes
...\otimes \mathcal{H}_{N}$ admit an $S_{1}\times ...\times S_{N}$-setting
LHV description. Then:\medskip \newline
(i) $\rho $ admits any $K_{1}\times ...\times K_{N}$-setting LHV description
where $K_{1}\leq S_{1},$ $...,$ $K_{N}\leq S_{N};$\medskip \newline
(ii) for any sites $1\leq n_{1}<...<n_{M}\leq N,$ where $1\leq M<N,$ the
reduced $M$-partite state $\rho _{(n_{1},...,n_{M})}$ on $\mathcal{H}%
_{n_{1}}\otimes ...\otimes \mathcal{H}_{n_{M}}$ admits the $S_{n_{1}}\times
...\times S_{n_{_{M}}}$-setting LHV description.
\end{proposition}

We stress that an $N$-partite quantum state $\rho ,$ admitting the $%
K_{1}\times ...\times K_{N}$-setting LHV description, does not need to admit
an $S_{1}\times ...\times S_{N}$-setting LHV description with $%
S_{1}>K_{1},...,S_{N}>K_{N}.$

\begin{definition}
An $N$-partite quantum state $\rho $ is said to admit an LHV model of
Werner's type if \emph{any} setting\emph{\ }family of\emph{\ }EPR local $N$%
-partite joint measurements performed on this state admits one and the same
LHV model formulated by definition 4.
\end{definition}

Any separable state admits an LHV model of Werner's type. For a bipartite
case, this model is specified by (\ref{74}). The nonseparable Werner state
[8] $W_{d,\Phi }$ on $\mathbb{C}^{d}\otimes \mathbb{C}^{d},$ $d\geq 2,$ with
parameter $\Phi \geq -1+\frac{d+1}{d^{2}},$ admits [8] an LHV model of
Werner's type under any projective measurements of two parties.

From definitions 6, 7 it follows that if an $N$-partite quantum state $\rho $
admits an LHV model of Werner's type then it admits an LHV description for
any setting $S_{1}\times ...\times S_{N}.$ However, the converse of this
statement is not true and even if an $N$-partite quantum state $\rho $
admits an LHV description for any setting $S_{1}\times ...\times S_{N}$,
this does not imply that this $\rho $ admits an LHV model of Werner's type -
since for each concrete setting $S_{1}\times ...\times S_{N}$, a probability
space may depend not only a state $\rho $ but also on performed measurements.

From definition 6 and proposition 2 in section 3 it follows the following
statement.

\begin{proposition}
An arbitrary $N$-partite quantum state $\rho $ admits an $S_{1}\times 
\underset{N-1}{\underbrace{1\times ...\times 1}}$-setting LHV description
for any $S_{1}\geq 2.$
\end{proposition}

Consider now a convex combination of $N$-partite quantum states admitting an
LHV\ description for a definite $S_{1}\times ...\times S_{N}$-setting.

\begin{proposition}
Let each of quantum states $\rho _{1},...,\rho _{M}$ on $\mathcal{H}%
_{1}\otimes ...\otimes \mathcal{H}_{N}$ admit an $S_{1}\times ...\times
S_{N} $-setting LHV description. Then any their convex combination: 
\begin{equation}
\sum_{m}\gamma _{m}\rho _{m},\text{ \ \ \ }\gamma _{m}\geq 0,\text{ \ \ \ }%
\tsum\limits_{m}\gamma _{m}=1,  \label{80}
\end{equation}%
also admits the $S_{1}\times ...\times S_{N}$-setting LHV description.
\end{proposition}

\begin{proof}
Suppose that every state $\rho _{m}$ on $\mathcal{H}_{1}\otimes ...\otimes 
\mathcal{H}_{N}$ admits an $S_{1}\times ...\times S_{N}$-setting LHV
description. Then, by definition 6 and theorem 1, for any $S_{1}\times
...\times S_{N}$-setting family of $N$-partite joint measurements (\ref{24}%
), performed on $\rho _{m}$ and specified by POV measures $\mathrm{M}%
_{n}^{(s_{n})},$ $\forall s_{n},$ $\forall n,$ there exists a joint
probability distribution 
\begin{equation}
\mu _{m}(\mathrm{d}\lambda _{1}^{(1)}\times ...\times \mathrm{d}\lambda
_{1}^{(S_{1})}\times ...\times \mathrm{d}\lambda _{N}^{(1)}\times ...\times 
\mathrm{d}\lambda _{N}^{(S_{N})}),  \label{81}
\end{equation}%
returning all 
\begin{eqnarray}
&&P_{(s_{1},...,s_{N})}(\mathrm{d}\lambda _{1}^{(s_{1})}\times ...\times 
\mathrm{d}\lambda _{N}^{(s_{_{N}})}\text{ }|\text{ }\rho _{m})  \label{82} \\
&=&\mathrm{tr}[\rho _{m}\{\mathrm{M}_{1}^{(s_{1})}(\mathrm{d}\lambda
_{1}^{(s_{1})})\otimes ...\otimes \mathrm{M}_{N}^{(s_{_{N}})}(\mathrm{d}%
\lambda _{N}^{(s_{_{N}})})\}],\text{ \ }%
s_{1}=1,...,S_{1},...,s_{N}=1,...,S_{N},  \notag
\end{eqnarray}%
as marginals. This implies that, for a mixture $\eta =\sum_{m}\gamma
_{m}\rho _{m},$ every 
\begin{eqnarray}
&&P_{(s_{1},...,s_{N})}(\mathrm{d}\lambda _{1}^{(s_{1})}\times ...\times 
\mathrm{d}\lambda _{N}^{(s_{_{N}})}\text{ }|\text{ }\eta )  \label{83} \\
&=&\dsum\limits_{m}\gamma _{m}\mathrm{tr}[\rho _{m}\{\mathrm{M}%
_{1}^{(s_{1})}(\mathrm{d}\lambda _{1}^{(s_{1})})\otimes ...\otimes \mathrm{M}%
_{N}^{(s_{_{N}})}(\mathrm{d}\lambda _{N}^{(s_{_{N}})})\}]  \notag
\end{eqnarray}%
constitutes the corresponding marginal of distribution $\sum_{m}\gamma
_{m}\mu _{m}.$ Therefore, by item (c) of theorem 1, any $S_{1}\times
...\times S_{N}$-setting family of $N$-partite joint measurements on state $%
\sum_{m}\gamma _{m}\rho _{m}$ admits an LHV model. By definition 6, the
latter means that state $\eta _{\beta }$ admits the $S_{1}\times ...\times
S_{N}$-setting LHV description.
\end{proof}

\bigskip

In the following statement, proved in appendix, we establish a threshold
bound for an arbitrary noisy bipartite state to admit an $S_{1}\times S_{2}$%
-setting LHV description. In an $S_{1}\times 1$-setting (or $1\times S_{2}$%
-setting) case, this bound is consistent with the statement of proposition 4.

\begin{proposition}
Let a bipartite quantum state $\rho $ on $\mathbb{C}^{d_{1}}\otimes \mathbb{C%
}^{d_{2}},$ $d_{1},d_{2}\geq 2,$ do not admit the LHV description for a
given setting $S_{1}\times S_{2}$. The noisy state 
\begin{equation}
\eta _{\rho }(\gamma )=(1-\gamma )\frac{I_{\mathbb{C}^{d_{1}}\otimes \mathbb{%
C}^{d_{2}}}}{d_{1}d_{2}}+\gamma \rho ,\text{ \ \ \ \ \ }0\leq \gamma \leq
(1+\beta _{\rho })^{-1},  \label{84}
\end{equation}%
admits the $S_{1}\times S_{2}$-setting LHV description under any \emph{%
generalized} EPR local quantum measurements of two parties. In (\ref{84}), 
\begin{equation}
\beta _{\rho }=\min \left\{ d_{1}(S_{2}-1)||\tau _{\rho }^{(1)}||;\text{ }%
d_{2}(S_{1}-1)||\tau _{\rho }^{(2)}||\right\}  \label{85}
\end{equation}%
and $||\tau _{_{\rho }}^{(1)}||,$ $||\tau _{\rho }^{(2)}||$ are operator
norms of the reduced states $\tau _{\rho }^{(1)}=\mathrm{tr}_{\mathbb{C}%
^{d_{2}}}[\rho ]$ and $\tau _{\rho }^{(2)}=\mathrm{tr}_{\mathbb{C}%
^{d_{1}}}[\rho ]$ on $\mathbb{C}^{d_{1}}$ and $\mathbb{C}^{d_{2}},$
respectively.
\end{proposition}

\bigskip

As an example, let us specify bound (\ref{84}) for the noisy state 
\begin{equation}
\eta _{\psi }^{(d)}(\gamma )=(1-\gamma )\frac{I_{\mathbb{C}^{d}\otimes 
\mathbb{C}^{d}}}{d^{2}}+\gamma |\psi \rangle \langle \psi |,  \label{86}
\end{equation}%
on $\mathbb{C}^{d}\otimes \mathbb{C}^{d},$ $d\geq 2,$ induced by the
maximally entangled pure state $\psi =\frac{1}{\sqrt{d}}\sum_{m=1}^{d}e_{m}%
\otimes e_{m}$, where $\{e_{m}\}$ is an orthonormal basis in $\mathbb{C}%
^{d}. $

In this case, $||\tau _{|\psi \rangle \langle \psi |}^{(n)}||=\frac{1}{d},$ $%
n=1,2,$ and substituting this into (\ref{84}), we conclude that state $\eta
_{|\psi \rangle \langle \psi |}^{(d)}(\gamma )$ admits an $S_{1}\times S_{2}$%
-setting LHV description under any \emph{generalized }quantum measurements
of two parties whenever%
\begin{equation}
0\leq \gamma \leq \frac{1}{1+\underset{n=1,2}{\min }(S_{n}-1)}.  \label{87}
\end{equation}%
Note that the partial transpose of $\eta _{|\psi \rangle \langle \psi
|}^{(d)}(\gamma )$ has the eigenvalue $\frac{1-\gamma (d+1)}{d^{2}}$, which
is negative for any $\gamma >\frac{1}{d+1}.$ Therefore, due to the Peres
separability criterion [21], state $\eta _{|\psi \rangle \langle \psi
|}^{(d)}(\gamma )$ is nonseparable for any $\gamma \in (\frac{1}{d+1},1].$
Thus, for state (\ref{86}), bound (\ref{87}) is nontrivial whenever $%
\min_{n=1,2}(S_{n}-1)$ $<d.$

\section{Conclusions}

In the present paper, we introduce a general framework for the probabilistic
description of a multipartite correlation scenario with an arbitrary number
of settings and any spectral type of outcomes at each site. This allows us:

\begin{enumerate}
\item[$\bullet $] To specify in probabilistic terms the difference between
nonsignaling [17], the EPR locality [1] and Bell's locality [4, 5] and to
show that, in contrast to the opinion of Bell [3, 5]:

(i) the EPR paradox [1] cannot be, in principle, resolved via the violation
of Bell's locality since the latter type of locality is only sufficient but
not necessary for the type of locality meant by Einstein, Podolsky and Rosen
in [1] - the EPR locality;

(ii) the EPR locality is not violated\footnote{%
See also our discussion in [12].} under a multipartite correlation
experiment on spacelike separated quantum particles and the so-called
"quantum nonlocality" does not constitute the violation of locality of
quantum interactions;

\item[$\bullet $] To introduce the notion of an LHV model for an $%
S_{1}\times ...\times S_{N}$-setting $N$-partite correlation experiment with
outcomes of any spectral type, discrete or continuous, and to stress that
the same correlation experiment may admit several LHV models and that the
existence of an LHV model of a general type is necessarily linked with only
nonsignaling but does not need to imply the EPR locality and even Bell's
locality;

\item[$\bullet $] To prove general statements on an LHV simulation of an
arbitrary $S_{1}\times ...\times S_{N}$-setting $N$-partite correlation
experiment. These statements not only generalize to an arbitrary
multipartite case, with outcomes of any spectral type, discrete or
continuous, the necessary and sufficient conditions introduced by Fine [7]
for a $2\times 2$-setting case, with two outcomes per site, but also
establish the equivalence between the existence of an LHV model for joint
probability distributions and the existence of the LHV-form representation
for the product expectations of the special type;

\item[$\bullet $] To introduce the notion of an $N$-partite quantum state
admitting an $S_{1}\times ...\times S_{N}$-setting LHV description; to prove
the main general statements on this notion and to establish its relation to
Werner's concept [8] of an LHV model for a multipartite quantum state;

\item[$\bullet $] To evaluate a threshold visibility for an arbitrary noisy
bipartite quantum state to admit an $S_{1}\times S_{2}$-setting LHV
description.
\end{enumerate}

In the sequel [25] to this paper, for an $S_{1}\times ...\times S_{N}$%
-setting $N$-partite correlation experiment with outcomes of any spectral
type, discrete or continuous, we introduce a single general representation
incorporating in a unique manner all Bell-type inequalities (on either joint
probabilities or correlation functions) that have been introduced in the
literature ever since the seminal publication [4] of Bell on the original
Bell inequality.

\section{Appendix}

Consider the proof of proposition 6 in section 5. For the $2\times 2$%
-setting case, this proof is similar to our proof of theorem 1 in [22].

According to definition 6, in order to prove that state $\eta _{\rho
}(\gamma )$ admits an $S_{1}\times S_{2}$-setting LHV description, we need
to show that any $S_{1}\times S_{2}$-setting family of bipartite joint
quantum measurements performed on $\eta _{\rho }(\gamma )$ admits an LHV
model.

Let, at each site, quantum measurements be described by POV measures $%
\mathrm{M}_{1}^{(s_{1})}(\mathrm{d}\lambda _{1}^{(s_{1})}),$ $%
s_{1}=1,...,S_{1},$ and $\mathrm{M}_{2}^{(s_{_{2}})}(\mathrm{d}\lambda
_{2}^{(s_{_{2}})}),$ $s_{2}=1,...,S_{2}.$ From formula (\ref{24}) it follows
that distributions $P_{(s_{1},.s_{2})}(\cdot |\eta _{\rho })$ have the form%
\begin{eqnarray}
P_{(s_{1},.s_{2})}(\mathrm{d}\lambda _{1}^{(s_{1})}\times \mathrm{d}\lambda
_{2}^{(s_{_{2}})}\text{ }|\text{ }\eta _{\rho }) &=&\mathrm{tr}[\eta _{\rho
}\{\mathrm{M}_{1}^{(s_{1})}(\mathrm{d}\lambda _{1}^{(s_{1})})\otimes \mathrm{%
M}_{2}^{(s_{_{2}})}(\mathrm{d}\lambda _{2}^{(s_{_{2}})})\}],  \TCItag{A1} \\
s_{1} &=&1,..,S_{1},\text{ \ \ \ }s_{2}=1,...,S_{2}.  \notag
\end{eqnarray}%
For state $\eta _{\rho }$ on $\mathbb{C}^{d_{1}}\otimes \mathbb{C}^{d_{2}},$
introduce self-adjoint operators $T_{\blacktriangleright }$ on $\mathbb{C}%
^{d_{1}}\otimes (\mathbb{C}^{d_{2}})^{\otimes S_{2}}$ and $%
T_{\blacktriangleleft }$ on $(\mathbb{C}^{d_{1}})^{\otimes S_{1}}\otimes 
\mathbb{C}^{d_{2}},$ satisfying the relations:%
\begin{eqnarray}
\mathrm{tr}_{\mathbb{C}^{d_{2}}}^{(k_{1},...,k_{S_{2}-1})}[T_{%
\blacktriangleright }] &=&\eta _{\rho },\ \ \ \ \ \ \ \ 2\leq
k_{1}<..<k_{S_{2}-1}\leq 1+S_{2},  \TCItag{A2} \\
\mathrm{tr}_{\mathbb{C}^{d_{1}}}^{(j_{1},...,j_{S_{1}-1})}[T_{%
\blacktriangleleft }] &=&\eta _{\rho },\text{ \ \ \ \ \ \ \ \ }1\leq
j_{1}<...<j_{S_{1}-1}\leq S_{1}.  \TCItag{A3}
\end{eqnarray}%
Here: (i) the lower indices of operators $T_{\blacktriangleright }$ and $%
T_{\blacktriangleleft }$ indicate a direction of extension of the Hilbert
space $\mathbb{C}^{d_{1}}\otimes \mathbb{C}^{d_{2}};$ (ii) $\mathrm{tr}_{%
\mathbb{C}^{d_{2}}}^{(k_{1},...,k_{S_{2}-1})}[\cdot ]$ denotes the partial
trace over elements of $\mathbb{C}^{d_{2}}$, standing in $k_{1}$-th, $...,$ $%
k_{S_{2}-1}$-th places in tensor products in $\mathbb{C}^{d_{1}}\otimes (%
\mathbb{C}^{d_{2}})^{\otimes S_{2}}$. Similarly, for the partial trace $%
\mathrm{tr}_{\mathbb{C}^{d_{1}}}^{(j_{1},...,j_{S_{1}-1})}[\cdot ].$

As we prove in [24], for any bipartite quantum state, dilations $%
T_{\blacktriangleright }$ and $T_{\blacktriangleleft }$ exist. In [23, 24],
we refer to these dilations as source operators for a bipartite state. Note
that any positive source operator is a density operator.

If, for state $\eta _{\rho }(\gamma ),$ there exist \emph{density source
operators} $T_{\blacktriangleright }$ and $T_{\blacktriangleleft },$ then
the probability measures 
\begin{equation}
\mathrm{tr}[T_{\blacktriangleright }\{\mathrm{M}_{1}^{(s_{1})}(\mathrm{d}%
\lambda _{1}^{(s_{1})})\otimes \mathrm{M}_{2}^{(1)}(\mathrm{d}\lambda
_{2}^{(1)})\otimes ...\otimes \mathrm{M}_{2}^{(S_{2})}(\mathrm{d}\lambda
_{2}^{(S_{2})})\}],\text{ \ \ }s_{1}=1,...,S_{1},  \tag{A4}
\end{equation}%
and 
\begin{equation}
\mathrm{tr}[T_{\blacktriangleleft }\{\mathrm{M}_{1}^{(1)}(\mathrm{d}\lambda
_{1}^{(1)})\otimes ...\otimes \mathrm{M}_{1}^{(S_{1})}(\mathrm{d}\lambda
_{1}^{(S_{1})})\otimes \mathrm{M}_{2}^{(s_{2})}(\mathrm{d}\lambda
_{2}^{(s_{2})})\}],\text{ \ \ \ }s_{2}=1,...,S_{2},  \tag{A5}
\end{equation}%
constitute, correspondingly, distributions 
\begin{equation}
\mu _{\blacktriangleright }^{(s_{1})}(\mathrm{d}\lambda _{1}^{(s_{1})}\times 
\mathrm{d}\lambda _{2}^{(1)}\times ...\times \mathrm{d}\lambda
_{2}^{(S_{2})}),\text{ \ \ }s_{1}=1,...S_{1},  \tag{A6}
\end{equation}%
and%
\begin{equation}
\mu _{\blacktriangleleft }^{(s_{2})}(\mathrm{d}\lambda _{1}^{(1)}\times
...\times \mathrm{d}\lambda _{1}^{(S_{1})}\times \mathrm{d}\lambda
_{2}^{(s_{2})}),\text{ \ \ \ }s_{2}=1,...S_{2},  \tag{A7}
\end{equation}%
specified in theorem 2 of section 4.

Therefore, finding for state $\eta _{\rho }(\gamma )$ of a density source
operator $T_{\blacktriangleright }$ (or $T_{\blacktriangleleft }$) will
prove the existence for this state of an $S_{1}\times S_{2}$-setting LHV
description.

For a state $\rho $ standing in (\ref{84}), consider its spectral
decomposition: 
\begin{equation}
\rho =\sum_{i}\alpha _{i}|\Psi _{i}\rangle \langle \Psi _{i}|,\text{ \ \ }%
\langle \Psi _{i},\Psi _{j}\rangle =\delta _{ij}\text{, \ \ \ }\forall
\alpha _{i}>0,\text{ \ \ }\sum_{i}\alpha _{i}=1.  \tag{A8}
\end{equation}%
Let 
\begin{equation}
\Psi _{i}=\sum_{k}\Phi _{k}^{(i)}\otimes f_{k},\ \text{\ \ \ \ \ \ }%
\sum_{k}\langle \Phi _{k}^{(i)},\Phi _{k}^{(j)}\rangle =\delta _{ij}, 
\tag{A9}
\end{equation}%
be the Schmidt decomposition of eigenvector $\Psi _{i}$ with respect to an
orthonormal basis $\{f_{k}\}$ in $\mathbb{C}^{d_{2}}.$ Substituting this
into (A8), we thus derive 
\begin{equation}
\rho =\sum_{k,l=1}^{d_{2}}\rho _{kl}\otimes |f_{k}\rangle \langle f_{l}|,%
\text{ \ \ \ \ }\rho _{kl}=\sum_{i}\alpha _{i}|\Phi _{k}^{(i)}\rangle
\langle \Phi _{l}^{(i)}|.  \tag{A10}
\end{equation}%
Operators $\rho _{kk}$ are positive with $\sum_{k}\mathrm{tr}[\rho _{kk}]=1.$
Note that $\tau _{\rho }^{(1)}=\sum_{k}\rho _{kk}$ \ is the state on $%
\mathbb{C}^{d_{1}}$ reduced from $\rho .$

Introduce on $\mathbb{C}^{d_{1}}\otimes (\mathbb{C}^{d_{2}})^{\otimes S_{2}}$
the self-adjoint operator 
\begin{eqnarray}
T_{\blacktriangleright }(\gamma ) &=&(1-\gamma )\frac{I_{\mathbb{C}%
^{d_{1}}}\otimes I_{\mathbb{C}^{d_{2}}}\otimes I_{_{\mathbb{X}}}}{%
d_{1}d_{2}^{S_{2}}}\text{ }+\text{ }\gamma \sum_{k,l}\rho _{kl}\otimes \frac{%
\{|f_{k}\rangle \langle f_{l}|\otimes I_{\mathbb{X}})\}_{sym}}{%
d_{2}^{S_{2}-1}}  \TCItag{A11} \\
&&-\gamma (S_{2}-1)\tau _{\rho }^{(1)}\otimes \frac{I_{\mathbb{C}%
^{d_{2}}}\otimes I_{_{\mathbb{X}}}}{d_{2}^{S_{2}}},  \notag
\end{eqnarray}%
where $\mathbb{X}:\mathbb{=}$ $(\mathbb{C}^{d_{2}})^{\otimes (S_{2}-1)}$;
operator $\{|f_{k}\rangle \langle f_{l}|\otimes I_{_{\mathbb{X}}}\}_{sym}$
on $(\mathbb{C}^{d_{2}})^{\otimes S_{2}}$ represents the symmetrization of $%
|f_{k}\rangle \langle f_{l}|\otimes I_{_{\mathbb{X}}}$ and operators $\rho
_{kl}$ on $\mathbb{C}^{d_{1}}$ are defined by (A10). It is easy to verify
that $T_{\blacktriangleright }^{^{(1,S_{2})}}(\gamma )$ satisfies condition
(A2) and, therefore, constitutes a source operator for state $\eta _{\rho
}(\gamma ).$

In order to find $\gamma $ for which operator $T_{\blacktriangleright
}(\gamma )$ is positive, let us evaluate the sum of the first and the third
terms standing in (A11). Note that, in view of (A10), the second term in
(A11) constitutes a positive operator.

Taking into account the relation%
\begin{equation}
-\left\Vert Y\right\Vert \text{ }I_{_{\mathcal{K}}}\leq Y\leq \left\Vert
Y\right\Vert \text{ }I_{_{\mathcal{K}}},  \tag{A12}
\end{equation}%
holding for any bounded quantum observable $Y$ on a Hilbert space $\mathcal{K%
},$ we derive:%
\begin{eqnarray}
&&(1-\gamma )\frac{I_{\mathbb{C}^{d_{1}}}\otimes I_{\mathbb{C}%
^{d_{2}}}\otimes I_{_{\mathbb{X}}}}{d_{1}d_{2}^{S_{2}}}-\gamma (S_{2}-1)\tau
_{\rho }^{(1)}\otimes \frac{I_{\mathbb{C}^{d_{2}}}\otimes I_{_{\mathbb{X}}}}{%
d_{2}^{S_{2}}}  \TCItag{A13} \\
&\geq &\left[ 1-\gamma (1+d_{1}(S_{2}-1)||\tau _{\rho }^{(1)}||)\right] 
\text{ }\frac{I_{\mathbb{C}^{d_{1}}}\otimes I_{\mathbb{C}^{d_{2}}}\otimes
I_{_{\mathbb{X}}}}{d_{1}d_{2}^{S_{2}}}.  \notag
\end{eqnarray}%
Therefore, the source operator $T_{\blacktriangleright }(\gamma )$ is
positive (i.e a density source operator) for any 
\begin{equation}
0\leq \gamma \leq \left( 1+d_{1}(S_{2}-1)||\tau _{\rho }^{(1)}||\right)
^{-1}.  \tag{A14}
\end{equation}%
Quite similarly, state $\eta _{\rho }(\gamma )$ has a density source
operator $T_{\blacktriangleleft }(\gamma )$ for any 
\begin{equation}
0\leq \gamma \leq \left( 1+d_{2}(S_{1}-1)||\tau _{\rho }^{(2)}||\right)
^{-1}.  \tag{A15}
\end{equation}%
From (A14) and (A15) it follows that state $\eta _{\rho }(\gamma )$ admits
an $S_{1}\times S_{2}$-setting LHV description for any 
\begin{equation}
0\leq \gamma \leq \left( 1+\min \left\{ d_{1}(S_{2}-1)||\tau _{\rho
}^{(1)}||;\text{ }d_{2}(S_{1}-1)||\tau _{\rho }^{(2)}||\right\} \right) . 
\tag{A16}
\end{equation}%
Note that $||\tau _{\rho }^{(1)}||\geq \frac{1}{d_{1}}$ and $||\tau _{\rho
}^{(2)}||\geq \frac{1}{d_{2}}$.

\end{document}